\documentclass[10pt,journal,compsoc]{IEEEtran}
\ifCLASSOPTIONcompsoc
  \usepackage[nocompress]{cite}
\else
  \usepackage{cite}
\fi  
\ifCLASSINFOpdf
   \usepackage[pdftex]{graphicx}
\else 
   \usepackage[dvips]{graphicx}
\fi  
\usepackage{subfigure}
\usepackage{graphicx}
\usepackage{epstopdf}
\usepackage{amsmath,amsfonts,amsthm,amssymb} 
\usepackage{etoolbox} 
\usepackage{url}
\newtheorem{thm}{Proposition}

\makeatletter
\newcommand\figcaption{\def\@captype{figure}\caption}
\newcommand\tabcaption{\def\@captype{table}\caption}
\newtheorem{corollary}{\bf Corollary}

\begin{document} 

\title{Inference-Based Similarity Search in Randomized Montgomery Domains for Privacy-Preserving Biometric Identification}

\author{Yi~Wang,~ %\IEEEauthorrefmark{1},~  
        Jianwu~Wan,~
        Jun~Guo,~
        Yiu-Ming~Cheung,~ 
        and~Pong~C~Yuen~ 
\IEEEcompsocitemizethanks{  
\IEEEcompsocthanksitem Y. Wang and J. Guo are with the School of Computer Science and Network Security, Dongguan University of Technology, China. \protect\\% <-this % stops an unwanted space
E-mail: \{wangyi, guojun\}@dgut.edu.cn.
\IEEEcompsocthanksitem J. Wan is with the Department of Computer Science, Changzhou University, China. E-mail: jianwuwan@gmail.com.
\IEEEcompsocthanksitem Y.-M. Cheung and P. C. Yuen are with the Department of Computer Science, Hong Kong Baptist University, Hong Kong SAR, China. \protect\\% <-this % stops an unwanted space
E-mail: \{ymc, pcyuen\}@comp.hkbu.edu.hk. 
% note need leading \protect in front of \\ to get a newline within \thanks as
% \\ is fragile and will error, could use \hfil\break instead.
%\thanks{Manuscript received April 19, 2005; revised August 26, 2015.}
}}
 
\markboth{Accepted by IEEE Trans. Pattern Anal. and Mach. Intell.}%
{Wang \MakeLowercase{\textit{et al.}}: Inference-based Search }

\IEEEtitleabstractindextext{
\begin{abstract}  
Similarity search is essential to many important applications and often involves searching at scale on high-dimensional data based on their similarity to a query. In biometric applications, recent vulnerability studies have shown that adversarial machine learning can compromise biometric recognition systems by exploiting the biometric similarity information. Existing methods for biometric privacy protection are in general based on pairwise matching of secured biometric templates and have inherent limitations in search efficiency and scalability. In this paper, we propose an inference-based framework for privacy-preserving similarity search in Hamming space. Our approach builds on an obfuscated distance measure that can conceal Hamming distance in a dynamic interval. Such a mechanism enables us to systematically design statistically reliable methods for retrieving most likely candidates without knowing the exact distance values. We further propose to apply Montgomery multiplication for generating search indexes that can withstand adversarial similarity analysis, and show that information leakage in randomized Montgomery domains can be made negligibly small. Our experiments on public biometric datasets demonstrate that the inference-based approach can achieve a search accuracy close to the best performance possible with secure computation methods, but the associated cost is reduced by orders of magnitude compared to cryptographic primitives.
\end{abstract}

% Note that keywords are not normally used for peerreview papers.
}

% make the title area
\maketitle

\IEEEdisplaynontitleabstractindextext 
\IEEEpeerreviewmaketitle

\IEEEraisesectionheading{\section{Introduction}\label{sec:introduction}}

\IEEEPARstart{B}{eyond} conventional applications in forensic investigations and border control, personal identification based on biometrics (e.g., face, iris, fingerprint) is also expanding in the private sector. For instance, hundreds of US healthcare organizations are using iris and palm-vein biometrics for  patient identification and authentication \cite{irishealth}. India's DCB Bank has introduced Aadhaar-based fingerprint-reading ATMs where customers can access their bank accounts using their Aadhaar number and biometric details instead of a PIN \cite{atm}. The service is required to connect with the Aadhaar server for authenticating the identity of a customer every time a transaction is initiated. Today's social media websites provide photo-sharing services and are expected to handle face search in a database on the order of millions of enrolled records every day \cite{Wang16face}. Such applications of biometric identification in general involve handling queries and searches at scale in a networked environment.

Biometric identification systems typically require \emph{one-to-many} comparisons by evaluating \emph{biometric similarity} between an input query and the database records in some retrieval space. Security analysis has shown that information leakage in storing and processing biometric data can lead to identity theft and adversarial tracking \cite{Nandakumar15sp,Bringer15}. In particular, \emph{adversarial machine learning} of biometric recognition is found possible in the retrieval space \cite{Biggio15sp}. Even if individual templates are secured in the database, it was shown that an adversary can still exploit the biometric similarity information to compromise system operations. For instance, the similarity information can be used to fabricate a biometric spoof in a hill-climbing attack \cite{Galbally10} or to perform a biometric equivalent of the dictionary attack \cite{Boult07}. Another example is in multibiometric systems where multiple biometric traits are combined to increase the recognition accuracy and the population coverage of large-scale personal identification. Without fabricating any fake trait, it is possible to evade multibiometric systems if the attacker seizes the genuine similarity distribution of one single trait \cite{Rodrigues09,Biggio16pami}. Security and privacy risks associated with the biometric similarity information signify the importance of protecting it in the process of biometric identification.  

Binary representations of biometric features are of growing interest for search applications as they can ensure fast matching operations in Hamming space \cite{Alice15,Lu15,Sun09}. To achieve high discriminative power and matching accuracy, binary representations of biometric features typically contain hundreds, if not thousands, of bits \cite{Lim13,Sun09}.  
Moreover, due to acquisition noise and other factors, biometric features are known to have inherently large intra-class variations between samples that are taken from the same trait. 
This indicates that the input query is not likely to find an identical match in the biometric database and the true match may have a relatively large \emph{distance}, i.e., the number of mismatching bits, to the query for similarity evaluation in Hamming space.  
Rather than exhaustively exploring the database by pairwise comparisons, nearest neighbour (NN) methods for long binary representations with a large search radius of interest are proposed for fast retrieval in Hamming space \cite{Norouzi14,Muja14HD,Alice15}. These methods are designed to reduce the matching complexity of a similarity search by using certain data structures, e.g., hash tables.

Without privacy protection in design, NN methods have built-in mechanisms that are vulnerable to information leakage. For example, search indexes are often generated directly from biometric features through \emph{distance-preserving} transformations \cite{Lu15,Alice11,Alice07,Sun09}. This allows an adversary who can access the data structures to analyse the search indexes and glean critical similarity information of the dataset. Another source of information leakage is at the retrieval stage. To increase the retrieval accuracy, NN methods generally require distance computations and comparisons between the query and matching candidates retrieved from the data structure. From the perspective of privacy protection, it is desirable to return the retrieval results in order without revealing their distance values \cite{Bringer13sp}.

Secure multiparty computation techniques are proposed in the context of \emph{one-to-one} verification to protect biometric similarity evaluations \cite{Bringer13sp,Bringer15,Barni15sp,CloudID15}. They are mostly based on combining cryptographic primitives such as homomorphic encryption and garbled circuits to process encrypted biometric data without the need of decrypting them. As distance measures are not preserved in the encrypted domain, adversarial learning can be effectively prevented. In our context of one-to-many identification, however, a similarity search must return multiple candidates in order according to their similarity measures to the query. The process involves not only intensive distance computations but also intensive distance comparisons that are especially cumbersome to perform in the encrypted domain \cite{Wu14}.  

In this paper, we propose an inference-based framework for privacy-preserving similarity search of biometric identities in Hamming space. We argue that, from the perspective of adversarial machine learning, an attacker must collect sufficient and reliable training data that characterise the underlying biometric feature space. This motivates us to design effective mechanisms to diffuse the feature information over piecewise obfuscated sub-hash codes, rather than concealing everything in the encrypted domain. In this way, we aim to facilitate information-theoretic privacy in the design of hash-based indexing methods without sacrificing too much of search accuracy and efficiency.

The main contributions of this paper are:   
\begin{itemize} 
\item We start with the Hamming-ball search problem and model it as binary channel estimation for which we derive the optimal decision rule based on the theory of binary hypothesis testing. This enables the design of an inference-based approach for detecting, with a high accuracy in the retrieval space, matching candidates within a predefined neighbourhood of the query without knowing the actual distance values.

\item As a key element of the binary channel, we present a distance obfuscation mechanism based on matching piecewise obfuscated binary codes for hash-based similarity search. We derive analytically the condition where it is guaranteed to have a substring collision for neighbouring pairs, and show how it gives rise to an obfuscated distance measure that conceals the Hamming distance in a dynamic interval. 
 
\item We show how the inference-based approach for Hamming-ball search can be extended to handle a test-based rank-ordered search that arranges the retrieved candidates in rank order without the need of comparing distance values. We further show how to perform an approximate NN search by choosing candidates based on the obfuscated distance measure with or without the need of test-based ranking.

\item We apply Montgomery multiplication to generate privacy-preserving search indexes that allow detection of substring collisions in randomized Montgomery domains. This enables the design of a search protocol that restricts all parties involved in the process to perform their respective computations in separate Montgomery domains for privacy protection.

\item We study the privacy-preserving strength of our search scheme in an information-theoretic approach. In particular, we quantify the privacy gain of indexing in randomized Montgomery domains, and analyze how difficult it is for an adversary to infer biometric identities from search indexes in the proposed indexing and retrieval process.
\end{itemize}
 
The remainder of this paper is organized as follows. Section \ref{sec:review} reviews the related work. Section \ref{sec:search} provides the details of the inference-based similarity search framework. Section \ref{sec:index} presents the design of index generation in randomized Montgomery domains and the description of the search protocol. Analysis of the privacy-preserving strength is presented in Section \ref{sec:analysis}. Performance evaluation is reported in Section \ref{sec:perform}. Finally, we draw conclusions in Section \ref{sec:smry}.

\section{Related Work}
\label{sec:review}
 
Biometric privacy study has been largely focused on enabling one-to-one matching without revealing the biometric features that characterize an individual, known as biometric template protection \cite{Nandakumar15sp,Rane13bio,Patel15sp}.
Depending on how the protected reference is generated and matched, template protection schemes can be classified into bio-cryptosystems and feature transformations \cite{Nandakumar15sp}. Bio-cryptosystems generate error correction codewords for an indirect matching of biometric templates \cite{Rane13bio}. They typically yield a yes/no decision for verification on one-to-one basis, which is not suitable for returning candidates in the order of their matching degrees. Feature transformation methods apply non-invertible functions to biometric templates \cite{Patel15sp}. The protected references are usually distance preserving to perform matching directly in the transformed space. 

Recently, secure signal processing methods have been introduced to protect biometric matching by concealing both the feature contents and similarity evaluation in the encrypted domain \cite{Cavoukian14,Bringer15,Barni15sp,Lagendijk13}. These techniques use cryptographic primitives as a wrapper of distance-related functions. Thus, a privacy-preserving similarity search can be decomposed into two steps: secure distance computation followed by minimum (distance) finding via oblivious transfer \cite{Rane13ANN}. As unauthorized information gathering is precluded in the encrypted domain, secrecy and accuracy can be ensured in pairwise matching for biometric identification. However, cryptographic primitives are generally associated with high computation costs and excessive communication overheads. For instance, minimum finding is very complex to perform in the encrypted domain. It demands pairwise comparison of all encrypted distance values and frequent interactions between the query and the database \cite{Wu14}. This makes it intrinsically difficult, if not impossible, to meet the efficiency and scalability requirements of search applications, especially when dealing with high-dimensional data, in the encrypted domain \cite{Wu14,Rane13ANN}. 

Other than search in the encrypted domain, the concept of search with reduced reference is proposed in privacy-preserving content-based information retrieval to protect the original content and accelerate the search simultaneously \cite{Wu14,Weng15tifs,Weng16tkde}. The basic idea therein is to enforce $k$-anonymity or $l$-diversity properties by raising the ambiguity level of a data collection \cite{Gertz08,Bakken04}. Techniques of this paradigm are mostly based on randomized embedding \cite{Rane13ANN} and in particular locality sensitive hashing (LSH) \cite{Weng15tifs,Jimenez2015smh,Aghasaryan13lsh,Boufounos11}.  
LSH performs approximate NN search by hashing similar items into the same bucket and those that are distant from one another into different buckets, respectively, with high probabilities.
However, LSH by itself does not guarantee privacy \cite{Rane13ANN,Aghasaryan13lsh,Boufounos11}. 
It requires all parties involved in a search to use the same random keys in generating hash codes. 
Moreover, to achieve a good precision in search, LSH-based algorithms usually require a large number of random hash functions and may increase privacy risks \cite{Aghasaryan14lsh}.

Privacy-enhanced variants of LSH are recently proposed by combining LSH with cryptographic or information-theoretic protocols \cite{Rane13ANN}. The privacy protection framework proposed in \cite{Weng15tifs} generates a partial query instance by omitting certain bits in one or more sub-hash values to increase the ambiguity of query information for the server. The hash values of retrieved candidates are returned to the client for refinement. 
The framework is extended in \cite{Weng16tkde} where partial encryption is performed on the hash code of each item to prevent an untrustworthy server from precisely linking queries and database records. In particular, the server uses the unencrypted part of each item for approximate indexing and search while the client uses the encrypted part for re-ranking the candidates received from the server. To limit the number of candidates sent to the client, the server performs a preliminary ranking based on the \emph{partial distance} computed from the unencrypted part. The trade-off between privacy and utility of a search at the server side is therefore controlled by the number of unencrypted bits. We call this approach ``LSH + partial distance'' for brevity and will use it as one baseline approach for performance comparison.

\section{Privacy-Preserving Similarity Search}
\label{sec:search}

In this section, we study the privacy-preserving similarity search problem in the framework of binary hypothesis testing. Section \ref{sec:search:problem} details how we model the Hamming-ball search problem as binary channel estimation. Section \ref{sec:search:NPtest} derives the decision rule for the binary hypothesis test. Section \ref{sec:search:mimp} presents the design of the distance obfuscation mechanism and the obfuscated distance measure. Section \ref{sec:search:rank} describes the test-based rank-ordered search scheme and discusses how to perform an approximate NN search based on the obfuscated distance measure.

\subsection{Inference-Based Hamming-Ball Search}
\label{sec:search:problem} 
 
Consider a set $\Omega$ of binary strings, each of which represents a biometric identity and is of length $D$ bits.
Note that $D$ is typically large for biometric data. 
The query $\mathbf{q}$ is also of length $D$ bits. As binary embedding is usually distance-preserving, we treat two biometric identities as similar if their Hamming distance is close.  

Any string $\mathbf{p}\in\Omega$ is considered to be an $r$-\emph{neighbour} of the query $\mathbf{q}$ if the Hamming distance between $\mathbf{p}$ and $\mathbf{q}$, denoted by $d(\mathbf{p}, \mathbf{q})$, satisfies $d(\mathbf{p}, \mathbf{q}) \leq r$. All strings $\mathbf{p}\in\Omega$ that satisfy $d(\mathbf{p}, \mathbf{q})\leq r$
constitute the set, denoted by $\mathcal{B}(\mathbf{q}; r)$, of $r$-neighbours of the query $\mathbf{q}$ within its neighbourhood of Hamming radius $r$.
This is known as the $r$-neighbour detection or Hamming-ball search problem \cite{Norouzi14}. A related problem, known as $k$-NN search, is to find $k$ records in $\Omega$ that are closest in Hamming distance to the query $\mathbf{q}$. We first resolve the $r$-neighbour detection problem, and later show how it can be extended to handle the $k$-NN search problem.
 
Our inference-based approach for $r$-neighbour detection requires a distance obfuscation mechanism to conceal the exact value of $d(\mathbf{p}, \mathbf{q})$ in a certain interval $[a, b]$, where both variables $a$ and $b$ take on integer values and satisfy $a\leq d(\mathbf{p}, \mathbf{q}) \leq b$. In this way, it is clear that $d(\mathbf{p}, \mathbf{q})>r$ if $r < a$ and $d(\mathbf{p}, \mathbf{q})< r$ if $r > b$. There is an uncertainty in comparing $d(\mathbf{p}, \mathbf{q})$ with $r$ when $a \leq r \leq b$. In this case, we consider that an estimate of $d(\mathbf{p}, \mathbf{q})$, denoted by $\hat{d}(\mathbf{p}, \mathbf{q})$, is randomly chosen from the interval $[a, b]$ with equal probability. Then,
\begin{equation}
  {\rm Pr}\{\hat{d}(\mathbf{p}, \mathbf{q})\leq r\} = \frac{r-a+1}{b-a+1} ~.
\label{eq:x}
\end{equation} 
For convenience, let
\begin{equation}
\pi \stackrel{\rm def}{=} \frac{r-a+1}{b-a+1}
\label{eq:pi}
\end{equation}
from which we have $1/(b-a+1) \leq \pi \leq 1$ for $a\leq r \leq b$, and the value of $\pi$ depends on $a$ and $b$ for any given~$r$. 

\begin{figure}[t]
\centering
\includegraphics[width=.8\linewidth]{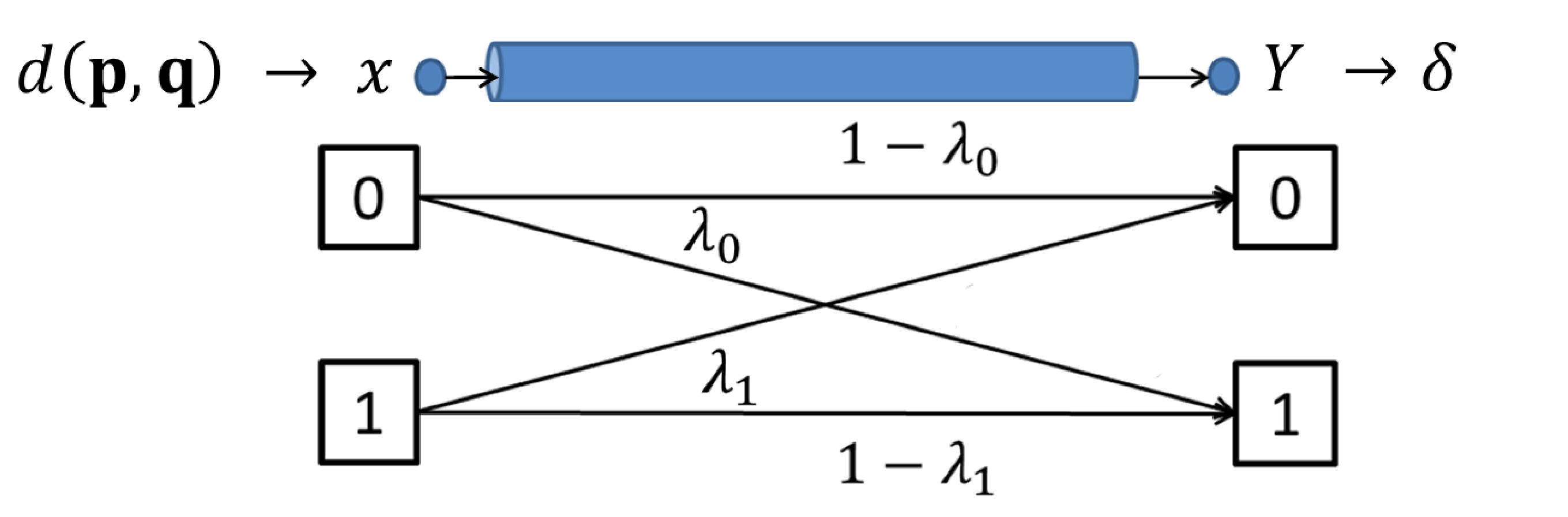}
\caption{Inference-based $r$-neighbour detection modelled as binary channel estimation with decision rule $\delta$.} 
\label{fig:bc}
\end{figure}  
  
We make inference based on the variable $\pi$. Specifically, let $x=0$ if $d(\mathbf{p}, \mathbf{q})> r$ and $x=1$ if $d(\mathbf{p}, \mathbf{q})\leq r$. Let $Y$ be a binary random variable taking on values $y=0$ if $\pi < \eta$ and $y=1$ if $\pi \geq \eta$, where $\eta \in(0, 1)$ is a predefined threshold. 
We resolve the $r$-neighbour detection problem as a binary hypothesis test $H_0: x=0, H_1: x=1$ that predicts the original binary input $x$ based on a decision rule $\delta$ as a function $\delta(y)$ of the values of $Y$.
  
This can be regarded as a binary channel estimation problem \cite{Poor94} where a binary digit $x$ is transmitted over some noisy channel that outputs a random observation $Y$ with outcome being either 0 or 1, as illustrated in Fig. \ref{fig:bc}.
Due to uncertainties introduced by distance obfuscation and threshold comparison, the original binary input $x$ may be received as $Y\neq x$ with probabilities $
\lambda_0 = {\rm Pr}\{Y=1\:|\:x=0\}$ and $\lambda_1 = {\rm Pr}\{Y=0\:|\:x=1\}$. 
We wish to find an optimal way to decide what was transmitted based on the values of $Y$. In the following section, we design the decision rule $\delta$ given the ``transmission'' error probabilities $\lambda_0$ and $\lambda_1$.

\subsection{Design of the Decision Rule $\delta$}
\label{sec:search:NPtest}

Any decision rule $\delta$ is associated with two types of errors in testing the \emph{null} hypothesis (i.e., $H_0$) versus the \emph{alternative} hypothesis (i.e., $H_1$).
Let $P_{\rm F}(\delta)$ denote the false-alarm probability that $H_1$ is falsely accepted when $H_0$ is true. Let $P_{\rm M}(\delta)$ denote the miss probability that $H_1$ is falsely rejected when $H_1$ is true. The detection probability $P_{\rm D}(\delta)=1-P_{\rm M}(\delta)$ is also called the \emph{power} of the decision rule $\delta$.  
There is a fundamental trade-off between $P_{\rm F}(\delta)$ and $P_{\rm M}(\delta)$ in the hypothesis test.
According to the Neyman-Pearson criteria \cite{Rice07}, the optimal decision rule is the one to the constrained optimization problem: $\max_\delta P_{\rm D}(\delta)$ subject to $P_{\rm F}(\delta)\leq \alpha$, with $\alpha$ being the \emph{significance level} of the test. The Neyman-Pearson lemma \cite{Rice07} states that the likelihood ratio test is the most powerful test.

Consider the binary channel estimation problem in Fig.~\ref{fig:bc}. The observation $Y$ has probability mass functions 
$p_x(y)=\lambda_x$ if $y\neq x$ and $p_x(y)=1-\lambda_x$ if $y=x$, for $x=0,1$. The likelihood ratio for an observation $Y=y$ is thus given by
\begin{equation} 
L(y) = \frac{p_{1}(y)}{p_{0}(y)} = \begin{cases}
\vspace{1mm}
\dfrac{\lambda_1}{1-\lambda_0},  & \text{if $y = 0$} \\
\dfrac{1-\lambda_1}{\lambda_0},  & \text{if $y = 1$} ~.
\end{cases}
\label{eq:Ly}
\end{equation} 
The Neyman-Pearson test rejects $H_0$ for $L(Y)>\eta$ for some $\eta$, and $\eta$ is chosen such that ${\rm Pr}\{L(Y)>\eta\}=\alpha$ if $H_0$ is true.   
When $\lambda_0+\lambda_1 < 1$, the Neyman-Pearson decision rule, which is the conditional probability of accepting $H_1$ given that we observe $Y=y$, turns out to be \cite{Poor94}
\begin{equation}
\tilde{\delta}_{\rm NP}(y) = \begin{cases}
\vspace{1mm}
\dfrac{\alpha}{\lambda_0}, & \text{if $y=1$} \\
0, & \text{if $y=0$} 
\end{cases}
\label{eq:alpha1}
\end{equation}
for $0\leq \alpha < \lambda_0$, and
\begin{equation}
\tilde{\delta}_{\rm NP}(y) = \begin{cases}
\vspace{1mm}
1, & \text{if $y=1$} \\
\dfrac{\alpha-\lambda_0}{1-\lambda_0}, & \text{if $y=0$}  
\end{cases} 
\label{eq:alpha2}
\end{equation}
for $\lambda_0\leq \alpha \leq 1$. 
The resulting detection probability is
\begin{equation}
P_{\rm D} (\tilde{\delta}_{\rm NP}) = \begin{cases}
\vspace{1mm}
\alpha\left(\dfrac{1-\lambda_1}{\lambda_0}\right), & \text{if $0\leq \alpha < \lambda_0$}   \\
1-\lambda_1 + \lambda_1\left(\dfrac{\alpha-\lambda_0}{1-\lambda_0}\right), & \text{if $\lambda_0 \leq \alpha < 1$}  ~.
\end{cases} 
\label{eq:Pd}
\end{equation}

It can be seen in (\ref{eq:alpha2}) that $\tilde{\delta}_{\rm NP}(y)=y$ when $\alpha=\lambda_0$. In other words, by allowing $P_{\rm F} (\delta) \leq \lambda_0$, the simple decision rule $\delta(y)=y$ that accepts $H_x$ if $y=x$ yields the same performance as the most powerful test according to the Neyman-Pearson criteria. The resulting detection probability is $P_{\rm D} (\delta) = 1 - \lambda_1$.
In this case, the inherent transmission error probability serves as an upper bound of the hypothesis testing error probability, i.e., $P_{\rm F} (\delta) + P_{\rm M} (\delta) \leq \lambda_0+\lambda_1$.  
This enables us to design the binary channel directly through a trade-off between the channel characteristics $\lambda_0$ and $\lambda_1$. 
 
\subsection{Design of the Binary Channel}
\label{sec:search:mimp}
 
As discussed in Section  \ref{sec:search:problem}, the binary channel in our context is in the form of a noisy channel that consists of distance obfuscation followed by threshold comparison. The former maps a distance value into an interval $[a, b]$ that defines $\pi$ in (\ref{eq:pi}). The latter yields $\lambda_0$ and $\lambda_1$ under the null and alternative hypotheses, respectively. Here, we present a distance obfuscation mechanism based on matching piecewise obfuscated binary codes in a carefully designed \emph{multi-index multi-probe} (MIMP) scheme for hash-based similarity search. 

Specifically, for all strings $\mathbf{p}\in\Omega$ and the query~$\mathbf{q}$, we divide each of them into $L$ non-overlapping segments, denoted by $\{\mathbf{p}^{(1)}, \mathbf{p}^{(2)}, \ldots, \mathbf{p}^{(L)}\}$ and $\{\mathbf{q}^{(1)}, \mathbf{q}^{(2)}, \ldots, \mathbf{q}^{(L)}\}$. By definition, we have $\mathbf{p}=\mathbf{p}^{(1)}||\mathbf{p}^{(2)}||\ldots||\mathbf{p}^{(L)}$ and $\mathbf{q}=\mathbf{q}^{(1)}|| \mathbf{q}^{(2)}||\ldots||\mathbf{q}^{(L)}$, where $||$ is the concatenation operator. If the string length $D$ is divisible by $L$, each substring is of length $s=D/L$ bits. If $D$ is not divisible by $L$, each of the first $(D\mod L)$ substrings is of length $s=\lceil D/L \rceil$ bits, and each of the remaining substrings has $s-1$ bits.  

For each substring $\mathbf{p}^{(i)}$ of $\mathbf{p}$, $i=1, 2, \ldots, L$, we create a set $\mathcal{V}_{\mathbf{p}^{(i)}}$ containing $\mathbf{p}^{(i)}$ and a one-bit variant of $\mathbf{p}^{(i)}$ that differs from $\mathbf{p}^{(i)}$ by one bit at a randomly selected position. For convenience, we write $\mathcal{V}_{\mathbf{p}^{(i)}}=\{\tilde{\mathbf{p}}_j^{(i)}\}$ where $j=0,1$, with $\tilde{\mathbf{p}}_0^{(i)}=\mathbf{p}^{(i)}$ and $\tilde{\mathbf{p}}_1^{(i)}$ being the one-bit variant of $\mathbf{p}^{(i)}$.
For each substring $\mathbf{q}^{(i)}$ of the query $\mathbf{q}$, $i=1, 2, \ldots, L$, we create a set that contains all possible one-bit variants of $\mathbf{q}^{(i)}$, denoted by $\mathcal{V}_{\mathbf{q}^{(i)}} = \{\tilde{\mathbf{q}}^{(i)}_k\}$, where $k=1, 2, \ldots, s$ if the substring is of length $s$ bits and $k=1, 2, \ldots, s-1$ if the substring has $s-1$ bits. We call a match between $\tilde{\mathbf{p}}_j^{(i)}\in \mathcal{V}_{\mathbf{p}^{(i)}}$ and $\tilde{\mathbf{q}}_k^{(i)}\in \mathcal{V}_{\mathbf{q}^{(i)}}$ as a \emph{substring collision}. Let $C(\mathcal{V}_{\mathbf{p}^{(i)}}, \mathcal{V}_{\mathbf{q}^{(i)}})$ denote the \emph{collision count} (i.e., the number of all possible substring collisions) 
between $\mathcal{V}_{\mathbf{p}^{(i)}}$ and $\mathcal{V}_{\mathbf{q}^{(i)}}$.

\begin{thm}[]
\label{thm:collision}   
The collision count between $\mathcal{V}_{\mathbf{p}^{(i)}}$ and $\mathcal{V}_{\mathbf{q}^{(i)}}$ satisfies $C(\mathcal{V}_{\mathbf{p}^{(i)}}, \mathcal{V}_{\mathbf{q}^{(i)}}) \leq 1$ for all $i$.
The case $C(\mathcal{V}_{\mathbf{p}^{(i)}}, \mathcal{V}_{\mathbf{q}^{(i)}})=1$ implies $d(\mathbf{p}^{(i)}, \mathbf{q}^{(i)}) \leq 2$. The case  $C(\mathcal{V}_{\mathbf{p}^{(i)}}, \mathcal{V}_{\mathbf{q}^{(i)}})=0$ implies $d(\mathbf{p}^{(i)}, \mathbf{q}^{(i)}) \geq 2$.
\end{thm}
\begin{proof}
The following cases hold for each $i=1, 2, \ldots, L$. 
\begin{itemize}
\item  If $d(\mathbf{p}^{(i)}, \mathbf{q}^{(i)})=0$, we have $C(\mathcal{V}_{\mathbf{p}^{(i)}}, \mathcal{V}_{\mathbf{q}^{(i)}})=1$, since $\tilde{\mathbf{p}}^{(i)}_0 \neq \tilde{\mathbf{q}}^{(i)}_k$ for all $k$ and there exists exactly one $\tilde{\mathbf{q}}_k^{(i)}\in \mathcal{V}_{\mathbf{q}^{(i)}}$ such that $\tilde{\mathbf{p}}_1^{(i)}= \tilde{\mathbf{q}}_k^{(i)}$.

\item  If $d(\mathbf{p}^{(i)}, \mathbf{q}^{(i)})=1$, we have $C(\mathcal{V}_{\mathbf{p}^{(i)}}, \mathcal{V}_{\mathbf{q}^{(i)}})=1$, since $\tilde{\mathbf{p}}^{(i)}_1 \neq \tilde{\mathbf{q}}^{(i)}_k$ for all $k$ and there exists exactly one $\tilde{\mathbf{q}}_k^{(i)}\in \mathcal{V}_{\mathbf{q}^{(i)}}$ such that $\tilde{\mathbf{p}}_0^{(i)}= \tilde{\mathbf{q}}_k^{(i)}$.
 
\item 
If $d(\mathbf{p}^{(i)}, \mathbf{q}^{(i)}) = 2$, we have $C(\mathcal{V}_{\mathbf{p}^{(i)}}, \mathcal{V}_{\mathbf{q}^{(i)}})\leq 1$ where the equality holds if and only if the one-bit variant of $\mathbf{p}^{(i)}$ flips at a bit position at which $\mathbf{p}^{(i)}$ and $\mathbf{q}^{(i)}$ differ. 

\item 
If $d(\mathbf{p}^{(i)}, \mathbf{q}^{(i)}) > 2$, we have $C(\mathcal{V}_{\mathbf{p}^{(i)}}, \mathcal{V}_{\mathbf{q}^{(i)}})=0$, since $\tilde{\mathbf{p}}^{(i)}_j \neq \tilde{\mathbf{q}}^{(i)}_k$ for all $j, k$.
\end{itemize}
\end{proof}
 
\begin{corollary}\label{corollary}
The sum of collision count  $C(\mathcal{V}_{\mathbf{p}^{(i)}}, \mathcal{V}_{\mathbf{q}^{(i)}})$ over all $i$ is at most $L$, i.e.,  $\sum_{i=1}^{L} C(\mathcal{V}_{\mathbf{p}^{(i)}}, \mathcal{V}_{\mathbf{q}^{(i)}}) \leq L$.
\end{corollary}

\begin{corollary}\label{corollary2}
Since $0 \leq d(\mathbf{p}^{(i)}, \mathbf{q}^{(i)})\leq s$, we have $0 \leq d(\mathbf{p}^{(i)}, \mathbf{q}^{(i)}) \leq 2$ given
$C(\mathcal{V}_{\mathbf{p}^{(i)}}, \mathcal{V}_{\mathbf{q}^{(i)}})=1$, and $2 \leq d(\mathbf{p}^{(i)}, \mathbf{q}^{(i)})\leq s$ given $C(\mathcal{V}_{\mathbf{p}^{(i)}}, \mathcal{V}_{\mathbf{q}^{(i)}})=0$.
\end{corollary}

\begin{thm}[]
\label{thm:sic}   
If $\mathbf{p} \in \mathcal{B}(\mathbf{q}; r)$ with $r < 2L$, then there exists at least one $i$, $1 \leq i \leq L$, such that $C(\mathcal{V}_{\mathbf{p}^{(i)}}, \mathcal{V}_{\mathbf{q}^{(i)}}) = 1$.  
\end{thm}
\begin{proof}
Suppose that $C(\mathcal{V}_{\mathbf{p}^{(i)}}, \mathcal{V}_{\mathbf{q}^{(i)}}) = 0$ for all $i$. By Proposition \ref{thm:collision}, we must have $d(\mathbf{p}^{(i)}, \mathbf{q}^{(i)})\geq 2$ for all $i$.
Then, $d(\mathbf{p}, \mathbf{q}) =  \sum_{i=1}^{L}d(\mathbf{p}^{(i)}, \mathbf{q}^{(i)})\geq 2L$. That is, we have $\mathbf{p} \notin \mathcal{B}(\mathbf{q}; r)$ with $r< 2L$, which gives a contradiction.
\end{proof}
 
Let $m$ denote the sum of collision count $C(\mathcal{V}_{\mathbf{p}^{(i)}},\mathcal{V}_{\mathbf{q}^{(i)}})$ over all $i$, i.e.,
\begin{equation}
m = \sum_{i=1}^L C(\mathcal{V}_{\mathbf{p}^{(i)}}, \mathcal{V}_{\mathbf{q}^{(i)}}) ~.
\label{eq:m}
\end{equation} 
Proposition \ref{thm:sic} suggests that, under the condition $r<2L$, the MIMP scheme ensures that any $r$-neighbor of the query $\mathbf{q}$ has a non-zero value of $m$. 
We shall see that $m$ serves as an obfuscated distance measure and is key to the design of the inference-based $r$-neighbor detection approach in our context. Note that $m\leq L$ from Corollary~\ref{corollary}. Without loss of generality, let us assume that 
$C(\mathcal{V}_{\mathbf{p}^{(i)}}, \mathcal{V}_{\mathbf{q}^{(i)}})=1$ for $i=1, 2, ..., m$ and $C(\mathcal{V}_{\mathbf{p}^{(i)}}, \mathcal{V}_{\mathbf{q}^{(i)}})=0$ for $i=m+1, ..., L$.
Then, by Corollary~\ref{corollary2}, we have $ 0 \leq \sum_{i=1}^{m} d(\mathbf{p}^{(i)}, \mathbf{q}^{(i)}) \leq 2m $ and $2(L-m)  \leq \sum_{i=m+1}^{L}  d(\mathbf{p}^{(i)}, \mathbf{q}^{(i)})  \leq s(L-m)$. 
Since by definition
 \begin{equation}
 d(\mathbf{p}, \mathbf{q}) = \sum_{i=1}^{m} d(\mathbf{p}^{(i)}, \mathbf{q}^{(i)}) +\sum_{i=m+1}^{L} d(\mathbf{p}^{(i)}, \mathbf{q}^{(i)})
 \end{equation}
we know that $d(\mathbf{p}, \mathbf{q})$ is within the range given by 
 \begin{equation}
 \label{eq:dm}
 2(L-m) \leq d(\mathbf{p}, \mathbf{q})  \leq s(L-m)+2m ~.
 \end{equation} 
Substituting $a= 2(L-m)$ and $b=s(L-m)+2m$ in (\ref{eq:pi}), we obtain
 \begin{equation}
 \pi   = \frac{2 - (2L - 1 - r)/m}{ (sL - 2L +1)/m   + 4-s}
      \label{eq:hashpi}  
 \end{equation}
for $2(L-m)  \leq r  \leq s(L-m)+2m$ or, equivalently, 
\begin{equation}
L-\frac{r}{2} \leq m \leq \frac{sL - r}{s-2} ~.
\label{eq:mrange1}
\end{equation}
Note that, with $r<2L$, we must have $(sL-r)/(s-2) > L$. Since $m\leq L$ by Corollary \ref{corollary}, we rewrite (\ref{eq:mrange1}) as 
\begin{equation}
L-\frac{r}{2} \leq m \leq L~.
\label{eq:mrange}
\end{equation}
In addition, with $s\geq 2$, $\pi$ in the form of \eqref{eq:hashpi} monotonically increases as $m$ increases. Accordingly, based on the definition of $\pi$, a larger value of $m$ corresponds to a higher probability that $\mathbf{p}$ is an $r$-neighbour of $\mathbf{q}$.

As a result of \eqref{eq:mrange}, any string $\mathbf{p}\in \Omega$ with $m< L-r/2$ is not an $r$-neighbour of the query $\mathbf{q}$. On the other hand, all strings $\mathbf{p}\in \Omega$ with $m$ being in the inference region defined by (\ref{eq:mrange}) are subjected to the hypothesis testing, from which we can obtain through \eqref{eq:hashpi} the empirical probability mass functions of $\pi$ under hypothesis $H_0$ and hypothesis $H_1$, respectively, and then choose a value for the threshold $\eta$ to set the channel characteristics $\lambda_0$ and $\lambda_1$.

\begin{figure}[t]
\centering
\includegraphics[width=.9\linewidth]{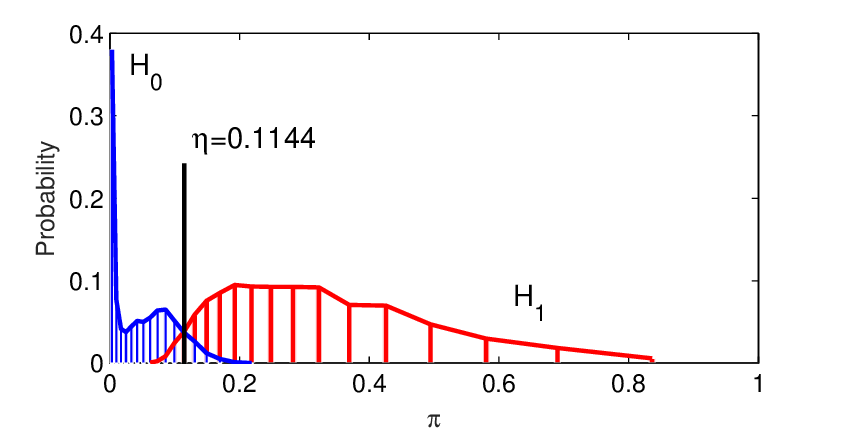}
\caption{Empirical probability mass functions of $\pi$ under $H_0$ and $H_1$, respectively, with the UBIRIS dataset, given $L=30$ and $r=50$. Choosing $\eta=0.1144$ minimizes $\lambda_0+\lambda_1$.}
\label{fig:pi} 
\end{figure}

For example, let us consider the UBIRIS dataset \cite{UBath} and use 400 eyes each with 10 iris images. Each iris image is represented by an iris code of 400 bits generated using a rotation invariant ordinal feature coding scheme \cite{Sun09}. We enroll each eye with one iris code at a time, and then use another iris code from the same eye to search. This is run for 24 times with different enrollments and queries, which results in a total of 3,840,000 pairwise comparisons, including 6,026 neighbouring pairs and 3,833,974 non-neighbouring pairs as ground truth, given the Hamming radius $r=50$.

We set $L=30$ and perform the MIMP scheme in each run. All enrolled strings whose $m$ value with respect to corresponding queries satisfying (\ref{eq:mrange}), which is $5\leq m \leq 30$ in this example, are subjected to inference. This involves 6,026 neighbouring pairs and 5,859 non-neighbouring pairs. Comparing with the ground truth, the inference set includes \emph{all} neighbouring pairs and a small fraction (i.e., 0.15 percent) of non-neighbouring pairs. By (\ref{eq:hashpi}), the $\pi$ values can be evaluated for these instances. Fig. \ref{fig:pi} plots the empirical probability mass functions of $\pi$ under $H_0$ and $H_1$, respectively.  

We choose the threshold $\eta$ to minimize $\lambda_0+\lambda_1$ for the binary channel, which we recall is the upper bound of the hypothesis testing error probability $P_{\rm F}(\delta) + P_{\rm M}(\delta)$ as a result of the chosen decision rule $\delta(y)=y$ discussed in Section \ref{sec:search:NPtest}. 
Note that we may also choose a smaller $\eta$ value to increase the detection probability $P_{\rm D}(\delta)$, i.e., $1-\lambda_1$, by allowing a larger $\lambda_0$. In Fig. \ref{fig:pi}, $\eta=0.1144$, which yields $\lambda_0=0.082$ and $\lambda_1 = 0.036$. 
Given the $\eta$ value, the inference-based approach determines if an enrolment is an $r$-neighbour of the query by testing if its $\pi$ value satisfies $\pi \geq \eta$. Accordingly, taking into account all the 3,840,000 pairwise comparisons in this example, the overall detection probability is 96.4 percent and the overall false-alarm probability is 0.0125 percent.

\subsection{Test-Based Rank-Ordered Search}
\label{sec:search:rank}

Conventionally, a rank-ordered search is done by ordering the retrieved candidates according to their distance values. In this section, we present a test-based rank-ordered search scheme without the need of comparing distance values. The basic idea is to apply our proposed inference-based $r$-neighbour detection approach and evaluate a series of {\rm nested} inference regions with sequentially increased values of the Hamming radius $r$.

Take the UBIRIS dataset for example. Let us consider three values for the Hamming radius $r$ as shown in Table~\ref{table:rank}, i.e., $r_1=30$, $r_2=40$, $r_3=50$. We perform the MIMP scheme with $L=30$. Thus, the inference region $\mathcal{I}$ defined by \eqref{eq:mrange} for each corresponding $r$ value is $\mathcal{I}_1=[15,30]$, $\mathcal{I}_2=[10,30]$, $\mathcal{I}_3=[5,30]$, respectively. Note that in the case of the MIMP scheme, as long as $L$ is fixed, the $m$ value of any enrollment remains the same regardless of the $r$ value. For each value of $r$, we obtain through \eqref{eq:hashpi} the empirical probability mass functions of $\pi$, and choose the threshold $\eta$ to minimize $\lambda_0+\lambda_1$. For the UBIRIS dataset, it turns out that $\eta_1=0.1064$, $\eta_2=0.1110$, $\eta_3=0.1144$, respectively.
 
\begin{table}[t] 
\centering
\caption{Inference-based $r$-neighbour detection over the UBIRIS dataset}
\label{table:rank} 
\renewcommand{\arraystretch}{1.1} 
\begin{tabular}[c]{|l|c|c|c|} 
\hline
Hamming radius $r$  & 30 & 40  & 50 \\
\hline
Inference region $\mathcal{I}$  & $[15, 30]$ & $[10, 30]$  & $[5, 30]$ \\
\hline
Threshold $\eta$ for $\pi$ &  0.1064 & 0.1110  & 0.1144 \\
\hline
Threshold $\mu$ for $m$  &  22  &  19   & 16  \\ 
\hline
Overall detection probability  & 96.5\% & 95.9\%  & 96.4\%  \\
\hline 
Overall false-alarm probability &  0.0186\% &  0.0140\%  &  0.0125\%  \\
\hline 
\end{tabular}
\end{table} 

Once the $\eta$ values are learned, we can determine if an enrollment $\mathbf{p}$ is an $r$-neighbour of the query $\mathbf{q}$ by testing if the $\pi$ value of $\mathbf{p}$ satisfies $\pi \geq \eta$. Given $\pi$ in the form of \eqref{eq:hashpi}, it can be shown that testing $\pi \geq \eta$ is equivalent to testing 
\begin{equation}\label{eq:mc}
m \geq \frac{\eta(sL-2L+1) + 2L-1-r}{2 + \eta(s-4)} \stackrel{\rm def}{=} \mu ~. 
\end{equation}
Note in \eqref{eq:mc} that it suffices for $\mu$ to assume integer values since $m$ in this context takes on integer values only. Thus, for the UBIRIS dataset, we have $\mu_1=22$, $\mu_2=19$, $\mu_3=16$, respectively. The resulting detection probability and false-alarm probability of $r$-neighbour detection are provided in Table~\ref{table:rank} for each value of $r$.

Given $\mu_1>\mu_2>\mu_3$ and since the inference region $\mathcal{I}$ is nested, i.e., $\mathcal{I}_1 \subset \mathcal{I}_2 \subset \mathcal{I}_3$, the inferred set $\mathcal{B}'(\mathbf{q}; r)$ of $r$-neighbours of the query $\mathbf{q}$ in this case must also be nested, i.e., $\mathcal{B}'(\mathbf{q}; r_1) \subset \mathcal{B}'(\mathbf{q}; r_2) \subset \mathcal{B}'(\mathbf{q}; r_3)$. In general, if we consider $G$ sequentially increased values of the Hamming radius $r$ and given that the inferred set of $r$-neighbours is nested, we can order the retrieved candidates in $G$ ranks. In particular, those in the set $\mathcal{B}'(\mathbf{q}; r_1)$ are designated as Rank 1 candidates. Subsequently, for $g=2,3,\ldots,G$, those in $\mathcal{B}'(\mathbf{q}; r_g)-\mathcal{B}'(\mathbf{q}; r_{g-1})$ form the set of Rank $g$ candidates.

As a special case, an approximate $k$-NN search can be performed by finding the smallest $G$ such that the set $\mathcal{B}'(\mathbf{q}; r_G)$ contains at least $k$ candidates. We return all candidates in the top $G-1$ ranks, and the remaining candidates are chosen from Rank $G$ candidates arbitrarily. The latter is due to the nature of binary hypothesis testing where candidates in the same rank are considered to be of the same similarity level to the query regardless of their $m$ values. It is interesting to observe that, in the case of the test-based rank-ordered search, given $\mu_i > \mu_j$, the $m$ values of Rank $i$ candidates must be larger than those of Rank $j$ candidates. This observation motivates us to further consider a simple implementation of $k$-NN search where we simply choose the top $k$ candidates in descending order of the $m$ value without the need of test-based ranking.

\section{Privacy-Preserving Index Generation}
\label{sec:index}
 
In Section \ref{sec:search}, we introduced an inference-based approach for similarity search. The MIMP scheme works by matching piecewise obfuscated binary codes and enables hash-based indexing of biometric identities. This can be regarded as \emph{substring obfuscation} by introducing controlled ``noise'' to piecewise biometric data before disseminating them over hash tables. However, as shown in Section \ref{sec:analysis:gain}, from the perspective of privacy protection, such segments of biometric data should not be used directly as search indexes and stored in hash tables. Cryptographic hashing may be used to generate secure signatures. However, standard cipher codes typically have hundreds of bits, e.g., 256 bits by SHA-256, while codes of length longer than 32 bits are not suitable for indexing data structures \cite{Norouzi14}. Subdivision of cipher codes does not help as Hamming distances are not preserved after encryption. In this paper, we address the problem by indexing in randomized Montgomery domains.

\subsection{Indexing in Montgomery Domains}
\label{sec:index:mont}

Montgomery multiplication has been used in cryptographic schemes mainly for implementing fast modulo operations with large-integer arithmetic \cite{Crandall01}. Here, we exploit its elementary form to generate randomized signatures. Specifically, let $N$ be a prime number, and let $R$ be a positive integer that is coprime to $N$, i.e., ${\rm gcd} (R, N)=1$.  
The \emph{Montgomery form} of an integer $x$ is defined as
\begin{equation}
M(x; R, N) = xR \mod N
\label{eq:MR}
\end{equation}  
where $N$ is the modulus and $R$ is the multiplier. We also call $M(x; R, N)$ as the \emph{$(R, N)$-residue} of $x$ following Definition 9.2.2 in \cite{Crandall01}. In our context, for a binary string $\mathbf{p}^{(i)}$ of length $s$ bits, its Montgomery form can be evaluated by the sum of  $(R, N)$-residues at each one-bit position in $\mathbf{p}^{(i)}$. That is,
\begin{equation}
\begin{split}
M(\mathbf{p}^{(i)}; R, N) & = \big(\sum_{b=1}^{s} \mathbf{p}^{(i)}[b] \cdot 2^{b-1} \big) R \mod N \\
         = \sum_{b=1}^{s} & \mathbf{p}^{(i)}[b] \cdot M(2^{b-1}; R, N)  \mod N 
\end{split}   
\label{eq:MR_p}
\end{equation}
where $\mathbf{p}^{(i)}[b]\in \{0,1\}$. The reference $M(2^{b-1}; R, N)$ may be pre-computed for large-integer arithmetic. We regard (\ref{eq:MR_p}) as a projection of  $\mathbf{p}^{(i)}$ into the Montgomery domain of $(R, N)$-residues, also called the $(R, N)$-domain for brevity. Note that Montgomery multiplication does not preserve the original distance values and thus the resulting signatures can withstand adversarial similarity analysis.

We consider a three-party scenario that involves a user, a data owner and a server, which is typical in a networked computing environment like cloud \cite{CloudID15}. The data owner is responsible for user enrolment, server registration and index generation. The server maintains the hash tables and provides the computation-intensive task of similarity search on behalf of the data owner. Upon a query from the user, the server retrieves the most likely candidates and returns the search results to the user. Fig. \ref{fig:proto} illustrates the main procedures of our three-party protocol that enables similarity search in randomized Montgomery domains. They are described below in more details. 
 
\begin{figure}[t]
\centering
\includegraphics[width=\linewidth]{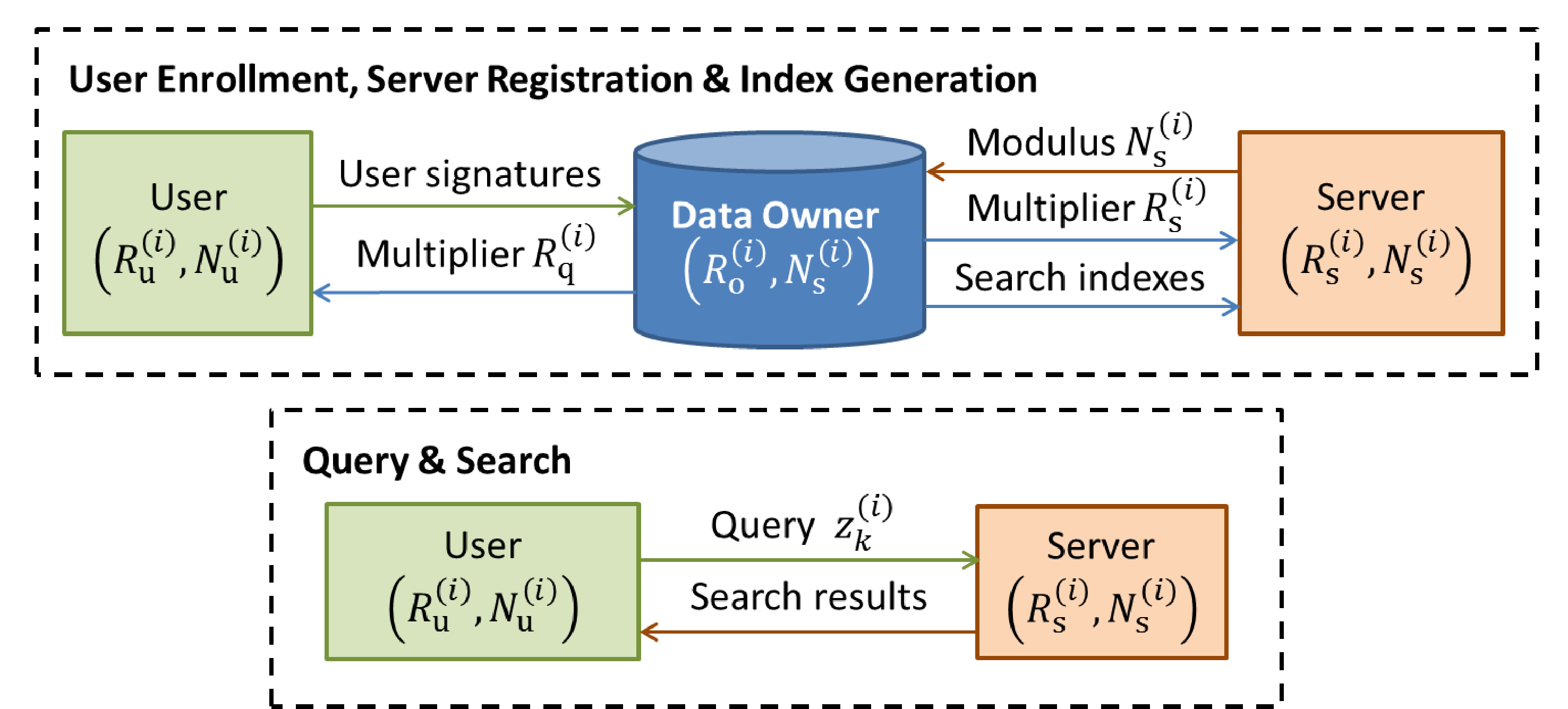}
\caption{Main procedures of the three-party search protocol.} 
\label{fig:proto} 
\end{figure} 

\textbf{User Enrolment:} The user generates the binary string~$\mathbf{p}$ from biometric data and divides it into $L$ substrings $\{\mathbf{p}^{(i)}\}$, each of which is associated with a set $\mathcal{V}_{\mathbf{p}^{(i)}}$ as described in Section \ref{sec:search:mimp}. For each chosen modulus $N^{(i)}_{\rm u}$, the user randomly generates a multiplier $R^{(i)}_{\rm u}$ that is coprime to $N^{(i)}_{\rm u}$, and maps all $\tilde{\mathbf{p}}_j^{(i)}\in\mathcal{V}_{\mathbf{p}^{(i)}}$ into the $(R^{(i)}_{\rm u}, N^{(i)}_{\rm u})$-domain. Then, the user sends $M(\tilde{\mathbf{p}}_j^{(i)}; R_{\rm u}^{(i)}, N^{(i)}_{\rm u})$ as the signature of $\tilde{\mathbf{p}}_j^{(i)}$ for all $i, j$ to the data owner.

\textbf{Server Registration:} 
For provision of service, the server registers at the data owner with a set of moduli $\{N^{(i)}_{\rm s}\}$. Then, for each $i$, the data owner randomly generates a multiplier $R_{\rm o}^{(i)}$ that is coprime to $N^{(i)}_{\rm s}$ for index generation.

\textbf{Index Generation:} The data owner maps the user's signature $M(\tilde{\mathbf{p}}_j^{(i)}; R_{\rm u}^{(i)}, N^{(i)}_{\rm u})$ into the $(R^{(i)}_{\rm o}, N^{(i)}_{\rm s})$-domain, and then sends the nested Montgomery form $M( M(\tilde{\mathbf{p}}_j^{(i)}; R^{(i)}_{\rm u}, N^{(i)}_{\rm u}); R_{\rm o}^{(i)}, N_{\rm s}^{(i)})$ for all $\tilde{\mathbf{p}}_j^{(i)}\in\mathcal{V}_{\mathbf{p}^{(i)}}$ to the server for indexing $\mathbf{p}^{(i)}$ in the $i$-th hash table. 
To enable the user to search at the server, the data owner randomly generates for each $i$ a multiplier $R^{(i)}_{\rm q}$ that is coprime to $N^{(i)}_{\rm s}$, and sends $R^{(i)}_{\rm q}$ for all $i$ to the user. Let $R'^{(i)}_{\rm q}$ be a \emph{modular multiplicative inverse} of $R^{(i)}_{\rm q}$ modulo $N^{(i)}_{\rm s}$, i.e.,
\begin{equation}
R^{(i)}_{\rm q} R'^{(i)}_{\rm q}\mod N^{(i)}_{\rm s} = 1 ~.
\label{eq:inverse}
\end{equation}
The data owner also computes
\begin{equation}
R^{(i)}_{\rm s} = M(R'^{(i)}_{\rm q}; R^{(i)}_{\rm o}, N^{(i)}_{\rm s})
\label{eq:Rs}
\end{equation}
and sends the resulting multiplier $R_{\rm s}^{(i)}$ for all $i$ to the server.

\textbf{Query and Search:} The user divides the query $\mathbf{q}$ into $L$ substrings and, for each substring $\mathbf{q}^{(i)}$, generates a set $\mathcal{V}_{\mathbf{q}^{(i)}}$ as described in Section \ref{sec:search:mimp}. For each $i$, the user maps all $\tilde{\mathbf{q}}^{(i)}_k \in \mathcal{V}_{\mathbf{q}^{(i)}}$ into the $(R^{(i)}_{\rm u}, N^{(i)}_{\rm u})$-domain, and computes 
\begin{equation}
z_k^{(i)} = M(\tilde{\mathbf{q}}_k^{(i)}; R^{(i)}_{\rm u}, N^{(i)}_{\rm u}) R_{\rm q}^{(i)}
\label{eq:zk}
\end{equation}
using the multiplier $R^{(i)}_{\rm q}$. The user sends $z_k^{(i)}$ to the server who in turn maps it into the $(R^{(i)}_{\rm s}, N^{(i)}_{\rm s})$-domain. Note that the result $M( z_k^{(i)}; R^{(i)}_{\rm s}, N^{(i)}_{\rm s})$ and the nested Montgomery form of $\tilde{\mathbf{q}}_k^{(i)}$, i.e., $M(M(\tilde{\mathbf{q}}_k^{(i)}; R^{(i)}_{\rm u}, N^{(i)}_{\rm u});  R^{(i)}_{\rm o}, N^{(i)}_{\rm s})$, are equivalent. This is because, by (\ref{eq:MR_p})-(\ref{eq:zk}) and using the congruence relation 
$a(b ~~{\rm mod}~ n) \equiv ab ~~{\rm mod}~  n ~~({\rm mod}~  n)$
that holds for all positive integers $a, b, n$, we have
\begin{equation}
\begin{split} 
M( &z_k^{(i)} ; R^{(i)}_{\rm s}, N^{(i)}_{\rm s}) = M( M(\tilde{\mathbf{q}}_k^{(i)}; R^{(i)}_{\rm u}, N^{(i)}_{\rm u}) R_{\rm q}^{(i)}; R^{(i)}_{\rm s}, N^{(i)}_{\rm s}) \\
=& M(\tilde{\mathbf{q}}_k^{(i)}; R^{(i)}_{\rm u}, N^{(i)}_{\rm u}) R^{(i)}_{\rm q} (R'^{(i)}_{\rm q} R^{(i)}_{\rm o} ~{\rm mod}~ N^{(i)}_{\rm s}) \mod N^{(i)}_{\rm s} \\
=& M(\tilde{\mathbf{q}}_k^{(i)}; R^{(i)}_{\rm u}, N^{(i)}_{\rm u}) R^{(i)}_{\rm q} R'^{(i)}_{\rm q} R^{(i)}_{\rm o} ~{\rm mod}~ N^{(i)}_{\rm s} \mod N^{(i)}_{\rm s} \\ 
=& M(\tilde{\mathbf{q}}_k^{(i)}; R^{(i)}_{\rm u}, N^{(i)}_{\rm u}) R^{(i)}_{\rm o} (R^{(i)}_{\rm q} R'^{(i)}_{\rm q} ~{\rm mod}~ N^{(i)}_{\rm s}) \mod N^{(i)}_{\rm s} \\
=& M(\tilde{\mathbf{q}}_k^{(i)}; R^{(i)}_{\rm u},  N^{(i)}_{\rm u}) R^{(i)}_{\rm o}  ~{\rm mod}~ N^{(i)}_{\rm s} \\
=& M(M(\tilde{\mathbf{q}}_k^{(i)}; R^{(i)}_{\rm u}, N^{(i)}_{\rm u});  R^{(i)}_{\rm o}, N^{(i)}_{\rm s})~.
\end{split} 
\label{eq:ybar}   
\end{equation}
Thus, the server can use the $(R^{(i)}_{\rm s}, N^{(i)}_{\rm s})$-residue of $z_k^{(i)}$ to probe the search indexes in the $i$-th hash table and perform the MIMP scheme as described in Section \ref{sec:search:mimp}. 

\subsection{Substring Collision in Montgomery Domains}
\label{sec:index:hash}

Note that, given a substring collision between $\tilde{\mathbf{p}}^{(i)}_j \in \mathcal{V}_{\mathbf{p}^{(i)}}$ and $\tilde{\mathbf{q}}^{(i)}_k \in \mathcal{V}_{\mathbf{q}^{(i)}}$, we must have 
\begin{equation}
\begin{split}
M(M(\tilde{\mathbf{p}}_k^{(i)}; & R^{(i)}_{\rm u}, N^{(i)}_{\rm u});  R^{(i)}_{\rm o}, N^{(i)}_{\rm s}) = \\
& M(M(\tilde{\mathbf{q}}_k^{(i)}; R^{(i)}_{\rm u}, N^{(i)}_{\rm u});  R^{(i)}_{\rm o}, N^{(i)}_{\rm s})~.
\end{split}
\label{eq:MM}
\end{equation}
However, having \eqref{eq:MM} does not necessarily imply $\tilde{\mathbf{p}}_j^{(i)} = \tilde{\mathbf{q}}_k^{(i)}$. Therefore, collision detection based on \eqref{eq:MM} may result in false positives but \emph{no} false negatives. In other words, matching in Montgomery domains may increase the false-alarm probability but have \emph{no} effect on the detection probability in our inference-based similarity search scheme. To reduce the false positives, we enhance the search protocol by indexing in Montgomery domains using \emph{multiple} independently generated signatures rather than one.
The idea can be conveniently explained by considering two arbitrary substrings of length $s$ bits as follows. 

Let $x$ and $y$ be the natural number representation of the two substrings, respectively. We first consider $T$ randomly chosen pairs of $(R_t, N_t)$, $t=1, 2, \ldots, T$, each yielding a Montgomery form for $x$ and $y$. For convenience, let $\gamma_t (x)=M(x; R_t, N_t)$ and $\gamma_t (y)=M(y; R_t, N_t)$, and let $\boldsymbol{\gamma}(x) = \{\gamma_t(x)\}$ and $\boldsymbol{\gamma}(y) = \{\gamma_t(y)\}$ denote the set of $T$ randomized signatures for $x$ and $y$, respectively. In this way, we consider the two substrings matched, i.e., $x=y$, if and only if $\boldsymbol{\gamma}(x)=\boldsymbol{\gamma}(y)$, i.e., $\gamma_t(x)=\gamma_t(y)$ for all $t$. 

Suppose that $R_t$ is encoded by $c_{\rm R}$ bits, i.e., $0<R_t < 2^{c_{\rm R}}$, and $N_t$ is encoded by $c_{\rm N}$ bits, i.e., $0<N_t < 2^{c_{\rm N}}$. Then, the binary code of $xR_t$ and $yR_t$ has $s+c_{\rm R}$ bits. Let $\nu(2^{c_{\rm N}})$ be the number of primes in the range $(0, 2^{c_{\rm N}})$. We choose each modulus $N_t$ independently from the $\nu(2^{c_{\rm N}})$ prime numbers at random with equal probability. A prime number is ``bad'' for matching $x$ and $y$ if the random choice of $N_t$ results in a false positive, i.e., $\gamma_t(x)=\gamma_t(y)$ given $x\neq y$. It is known that, for $s+c_{\rm R} < \nu(2^{c_{\rm N}})$, there can be \emph{at most} $s+c_{\rm R}-1$ such ``bad'' prime numbers in the range $(0, 2^{c_{\rm N}})$ \cite{Hromkovic05}. Since every prime number has an equal probability of being chosen, the probability $\Pr\{\gamma_t(x) = \gamma_t(y) \:|\: x\neq y \}$ is bounded by
\begin{equation}
\frac{s + c_{\rm R} -1}{\nu(2^{c_{\rm N}})}  \stackrel{\rm def}{=}  \beta(s; c_{\rm R}, c_{\rm N})  ~.
\label{eq:beta}
\end{equation}
The false positive probability $\Pr\{\boldsymbol{\gamma}(x) = \boldsymbol{\gamma}(y) \:|\: x \neq y \}$ is therefore bounded by $\beta(s; c_{\rm R}, c_{\rm N})^T$, which can be controlled by tuning the parameters $c_{\rm R}$, $c_{\rm N}$ and $T$ for a given $s$.

Next, consider the nested Montgomery form for $x$ and for $y$. That is, for each $\gamma_t (x)$ and $\gamma_t (y)$, we further choose $T$ pairs of $(R_v, N_v)$, $v=1, 2, \ldots, T$, independently at random. Each pair of $(R_v, N_v)$ yields the nested Montgomery form $M(\gamma_t (x); R_v, N_v)$ for $x$ and $M(\gamma_t (y); R_v, N_v)$ for $y$. For convenience, let $\psi_{tv}(x) = M(\gamma_t (x); R_v, N_v)$ and $\psi_{tv}(y) = M(\gamma_t (y); R_v, N_v)$.
Accordingly, we generate a set of $T^2$ randomized signatures for $x$ and $y$, denoted by $\boldsymbol{\psi}(x) = \{ \psi_{tv}(x)\}$ and $\boldsymbol{\psi}(y) = \{ \psi_{tv}(y)\}$, respectively. In this way, we consider $x=y$ if and only if $\boldsymbol{\psi}(x)=\boldsymbol{\psi}(y)$, i.e., $\psi_{tv}(x)=\psi_{tv}(y)$ for all $t$ and $v$. Since $N_t<2^{c_{\rm N}}$, both $\gamma_t (x)$ and $\gamma_t (y)$ can be encoded by $c_{\rm N}$ bits. Setting $s=c_{\rm N}$ in (\ref{eq:beta}) yields
$\Pr\{\psi_{tv} (x)=\psi_{tv} (y) \:|\: \gamma_{t} (x)\neq\gamma_t (y) \} \leq \beta(c_{\rm N}; c_{\rm R}, c_{\rm N})$. Then, it can be shown that $\Pr\{\psi_{tv} (x)=\psi_{tv} (y) \:|\: x \neq y \}$ is bounded by $\beta(s+c_{\rm N}; 2c_{\rm R}, c_{\rm N})$. The resulting false positive probability $\Pr \{ \boldsymbol{\psi}(x) = \boldsymbol{\psi}(y) \:|\:  x \neq y \}$ in the case of using nested Montgomery forms as search indexes is therefore bounded by $\beta (s+c_{\rm N}; 2 c_{\rm R}, c_{\rm N})^{T^2}$.

Fig. \ref{fig:beta} plots the bound $\beta (s+c_{\rm N}; 2 c_{\rm R}, c_{\rm N})^{T^2}$ with respect to the parameter $c_{\rm N}$ for $T=1, 2, 3$, given $s=14$ and $c_{\rm R} = 15$. It can be seen that the bound reduces drastically as $c_{\rm N}$ or $T$ increases. For example, with $T=2$, $c_{\rm N}=15$ and hence $\nu(2^{15})=3512$, the false positive probability is less than $7.44\times 10^{-8}$. In this paper, we use $T=2$ and $c_{\rm N}=15$ in our experimental settings. 

\begin{figure}[t] 
\centering
\includegraphics[width=.9\linewidth]{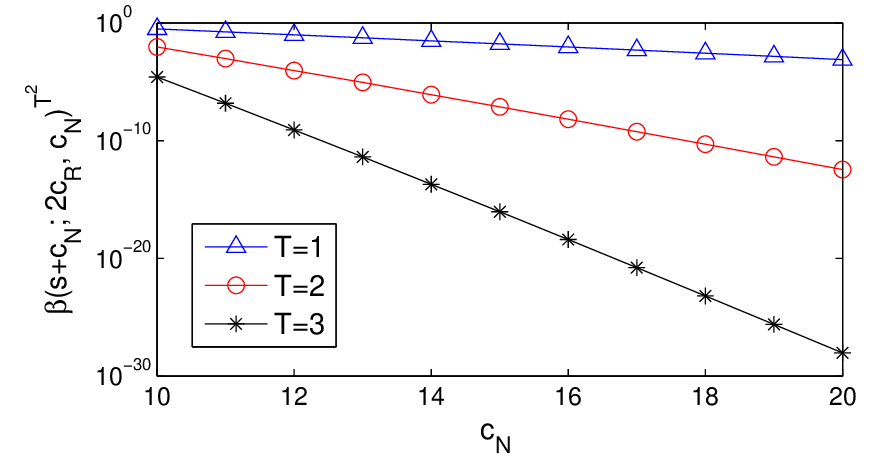}
\caption{Bound $\beta (s+c_{\rm N}; 2 c_{\rm R}, c_{\rm N})^{T^2}$ with respect to the modulus encoding length $c_{\rm N}$ for $T=1, 2, 3$, given $s=14$ and $c_{\rm R}=15$.}
\label{fig:beta} 
\end{figure}

\section{Privacy-Preserving Strength}
\label{sec:analysis}

In this section, we study the privacy-preserving strength of our inference-based framework for similarity search in randomized Montgomery domains. In particular, we are concerned with adversarial learning of the biometric similarity information that can be gleaned from the search indexes and the retrieval process. We begin by presenting in Section~\ref{sec:analysis:model} a security analysis of the three-party search protocol to understand what an attacker can do if they know the various primitives (i.e., multipliers and moduli) in each indexing step. We argue that the privacy threats may all boil down to the problem of how hard it is for the attacker to infer the biometric identity from the user's signatures in Montgomery forms. This motivates us to apply an information-theoretic approach for evaluating in Section~\ref{sec:analysis:entropy} the information leakage in Montgomery domains and in Section~\ref{sec:analysis:gain} the privacy gain in index generation. We further provide in Section~\ref{sec:analysis:tradeoff} a discussion on the well-known trade-off between privacy and utility that is applied also in our context of similarity search. 
 
\subsection{Security Analysis of the Search Protocol}
\label{sec:analysis:model}

In our design of the three-party search protocol, we consider a model where the different parties do not necessarily trust each other. We also assume that 1) the three parities do not collude with each other, 2) the communication channel between the parties is secure (e.g., encrypted), and 3) each party may behave in a curious-but-honest way, i.e., it follows the search protocol but may use any data in possession to glean additional information \cite{Weng16tkde,Weng15tifs,Biggio15sp,Bringer15}. 
As shown in Fig. \ref{fig:proto}, the three parties operate in separate Montgomery domains defined by the corresponding primitive pair, i.e., $(R_{\rm u}^{(i)}, N^{(i)}_{\rm u})$ for the user, $(R_{\rm o}^{(i)}, N^{(i)}_{\rm s})$ for the data owner, and $( R_{\rm s}^{(i)}, N^{(i)}_{\rm s})$ for the server. The following discusses the privacy threat in each main procedure of the search protocol.

\textbf{User Enrolment:}
The attacker may seize user signatures in the form of $M(\tilde{\mathbf{p}}_j^{(i)}; R_{\rm u}^{(i)}, N^{(i)}_{\rm u})$, but is not likely to know the user's primitive pair $(R_{\rm u}^{(i)}, N^{(i)}_{\rm u})$, as the latter is kept private to the user. In Section \ref{sec:analysis:entropy}, we show that recovering $\tilde{\mathbf{p}}_j^{(i)}$ from its Montgomery form without knowing the primitives can be made almost as difficult as a wild guess. On the other hand, even if the attacker is able to crack user signatures, it needs to do so for a sufficiently large number of users in order to derive the underlying biometric similarity distribution. This makes the inference cost of adversarial learning even more prohibitive.

\textbf{Server Registration:} 
We note that, if different application servers use the same set of moduli $\{N^{(i)}_{\rm s}\}$ to register provision of service at the data owner, the attacker can track users covertly by cross-matching over different applications, known as \emph{linkage attack} \cite{Nandakumar15sp}. To deal with this privacy threat, the servers can choose different sets of  moduli $\{N^{(i)}_{\rm s}\}$, so that the data owner can produce distinct sets of search indexes for the same biometric database.  

\textbf{Index Generation:} 
The attacker may sniff all the search indexes generated by the data owner. Since the search indexes are generated in the nested Montgomery form $M( M(\tilde{\mathbf{p}}_j^{(i)}; R^{(i)}_{\rm u}, N^{(i)}_{\rm u}); R_{\rm o}^{(i)}, N_{\rm s}^{(i)})$, the attacker cannot derive the biometric similarity distribution directly from the search indexes, as Montgomery multiplication does not preserve the original distance values. The attacker may alternatively attempt to glean the information by inferring each biometric identity from the search indexes and then performing pairwise comparisons for distance computation.
In this way, the attacker needs to compromise two layers of Montgomery multiplications defined in the respective Montgomery domains.
We note that secure communications can prevent the attacker from obtaining the primitives $N^{(i)}_{\rm s}$, $R^{(i)}_{\rm s}$ and $R^{(i)}_{\rm q}$. Otherwise, the attacker may derive the data owner's primitive $R^{(i)}_{\rm o}$, which will enable the attacker to obtain $M(\tilde{\mathbf{p}}_j^{(i)}; R_{\rm u}^{(i)}, N^{(i)}_{\rm u})$. The privacy threat is again reduced to cracking the user's signatures as discussed before. 

\begin{figure*}[t]
\centering
\begin{minipage}[t]{0.48\linewidth}
\centering
\includegraphics[width=\linewidth]{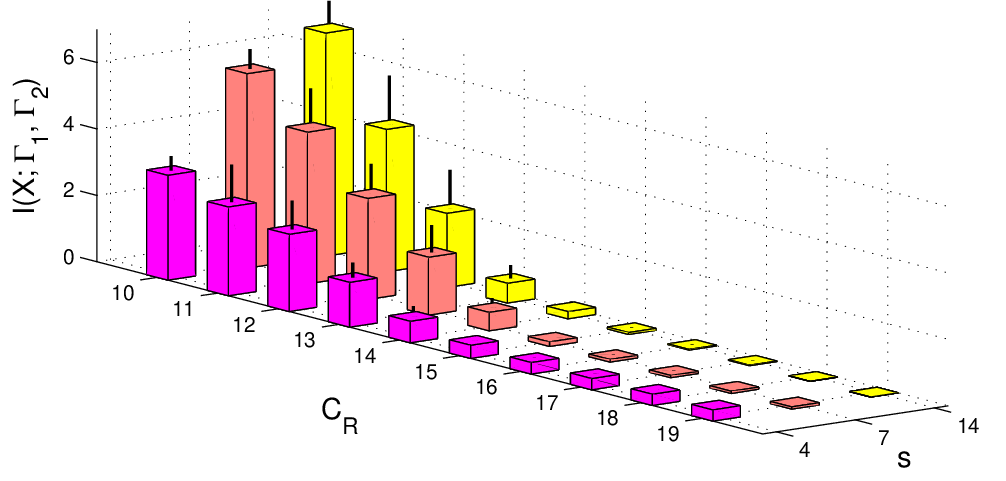} 
\caption{$I(X; \Gamma_1, \Gamma_2)$ with respect to the multiplier encoding length $c_{\rm R}$ for $s=4, 7, 14$. The mean and standard deviation are obtained with ten pairs of prime moduli $N_1$ and $N_2$ randomly chosen in $(0, 2^{15})$.}
\label{fig:entropy:sR}  
\end{minipage}
\hspace{6pt}
\begin{minipage}[t]{0.48\linewidth}
\centering
\includegraphics[width=\linewidth]{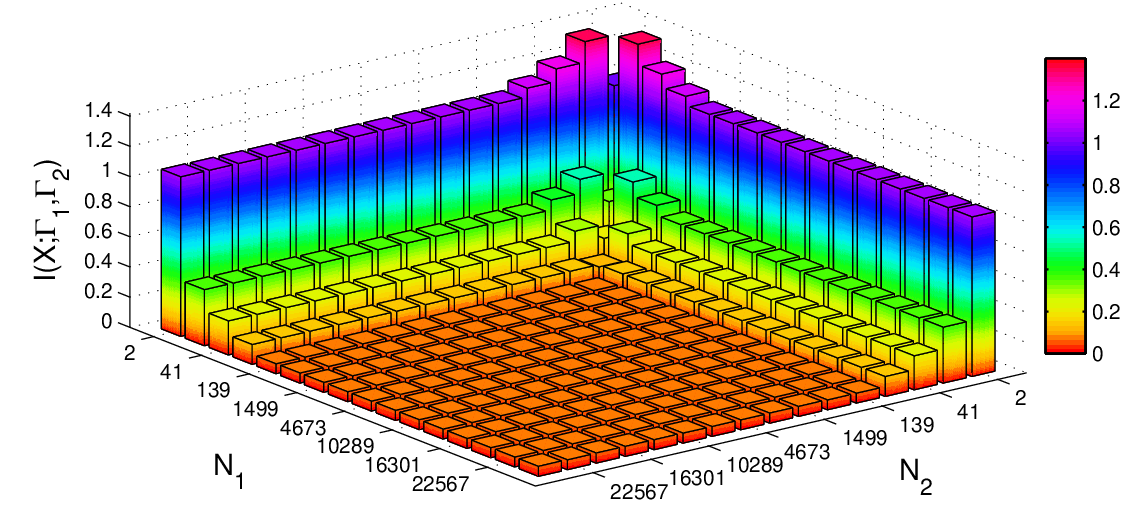}
\caption{$I(X; \Gamma_1, \Gamma_2)$ with respect to the prime moduli $N_1$ and $N_2$ chosen to span the range $(0, 2^{15})$, given $s=7$ and $c_{\rm R}=15$. 
}
\label{fig:entropy:N}  
\end{minipage}
\end{figure*}   

\textbf{Query and Search:} Recall that the obfuscated distance measure is key to our design of the inference-based similarity search scheme. It enables the server to retrieve most likely candidates without knowing the exact distance values. The attacker may attempt to glean the similarity information from the obfuscated distance measure. As a result, the risk of data disclosure in the retrieval process also depends on how different the distribution of the obfuscated distance values can be made from that of the exact distance values, which we will discuss in Section \ref{sec:analysis:tradeoff}.

\subsection{Information Leakage in Montgomery Domains}
\label{sec:analysis:entropy}
 
In privacy studies, \emph{mutual information} is often used to measure the average risk of data disclosure \cite{Sweeney02}. Accordingly, we quantify information leakage in our context as the mutual information between a substring data variable and its independently generated signatures in Montgomery domains. Specifically, consider a substring of length $s$ bits whose natural number representation is a discrete random variable $X$ with probability mass function $p(x) = 2^{-s}$ for all $x$ in the range $\mathcal{X}=[0, 2^{s})$. The entropy of $X$, denoted by $H(X)$, is thus $s$ bits \cite{Cover06}. For $t=1, 2$ and given the modulus $N_t$, the Montgomery form of $X$, denoted by $\Gamma_t$, is also a discrete random variable. In particular, let $\mathfrak{R}(N_t)$ be the set of all eligible values of the multiplier $R_t$ in the range $(0, 2^{c_{\rm R}})$ that are coprime to $N_t$. Then, for $X=x$, every $R_t\in\mathfrak{R}(N_t)$ yields an $(R_t, N_t)$-residue, i.e., $\gamma_t =  M(x; R_t, N_t)$. Let all possible values of $\Gamma_t$ constitute the subset $\mathcal{G}_t \subset [0, N_t)$. The mutual information between $X$ and its Montgomery forms $\Gamma_1, \Gamma_2$, denoted by $I(X; \Gamma_1, \Gamma_2)$, is then given by \cite{Cover06}
\begin{equation}
I(X; \Gamma_1, \Gamma_2) = H(X) - H(X|\Gamma_1, \Gamma_2) ~.
\label{eq:mutual}
\end{equation}
The conditional entropy $H(X|\Gamma_1, \Gamma_2)$ can be obtained by the \emph{chain rule for entropy} \cite{Cover06} as
\begin{equation}
H(X|\Gamma_1, \Gamma_2) = H(\Gamma_1, \Gamma_2, X) - H(\Gamma_1, \Gamma_2)
\label{eq:equivocation}
\end{equation}
where $H(\Gamma_1, \Gamma_2, X)$ and $H(\Gamma_1, \Gamma_2)$ can be obtained from the joint distribution $p(x, \gamma_1, \gamma_2)$ as follows. Since $\Gamma_1$ and $\Gamma_2$ are independent, we have 
\begin{equation}
p(x, \gamma_1, \gamma_2) = p(\gamma_1, \gamma_2 \:|\: x)p(x) = p(x) \prod_{t=1}^{2} p(\gamma_t\:|\:x) 
\end{equation}
where the conditional distribution $p(\gamma_t\:|\:x)$ can be derived by enumerating all $R_t\in\mathfrak{R}(N_t)$ with respect to $N_t$. The marginal distribution $p(\gamma_1, \gamma_2)$ can be obtained as 
\begin{equation} 
p(\gamma_1, \gamma_2) = \sum_{x\in\mathcal{X}} p(x, \gamma_1, \gamma_2) ~.
\end{equation} 
Then, the joint entropy $H(\Gamma_1, \Gamma_2, X)$ is given by
\begin{equation}
H(\Gamma_1, \Gamma_2, X) = 
-\sum_{x\in\mathcal{X}} \sum_{\gamma_1\in\mathcal{G}_1} \sum_{\gamma_2\in\mathcal{G}_2}  p(x, \gamma_1, \gamma_2) \log p(x, \gamma_1, \gamma_2)  
\label{eq:joint1}   
\end{equation}  
and the joint entropy $H(\Gamma_1, \Gamma_2)$ is given by
\begin{equation}
H(\Gamma_1, \Gamma_2) = 
- \sum_{\gamma_1\in\mathcal{G}_1} \sum_{\gamma_2\in\mathcal{G}_2}  p(\gamma_1, \gamma_2) \log p(\gamma_1, \gamma_2) ~.
\label{eq:joint2}  
\end{equation}    

Fig. \ref{fig:entropy:sR} plots $I(X; \Gamma_1, \Gamma_2)$ with respect to the parameter $c_{\rm R}$ for $s=4, 7, 14$. With $c_{\rm N}=15$, we randomly select ten pairs of moduli $N_1$ and $N_2$ from the $\nu(2^{15})=3512$ available primes. For each pair of $N_1$ and $N_2$, and for each particular value of $c_{\rm R}$ and $s$, we calculate $I(X; \Gamma_1, \Gamma_2)$ and present in Fig. \ref{fig:entropy:sR} the mean and standard deviation of $I(X; \Gamma_1, \Gamma_2)$. It can be seen that, as $c_{\rm R}$ increases, $I(X; \Gamma_1, \Gamma_2)$ monotonically decreases and approaches zero. The larger the value of $s$, the smaller is $I(X; \Gamma_1, \Gamma_2)$ when it converges.
 
Fig. \ref{fig:entropy:N} shows the effect of modulus on $I(X; \Gamma_1, \Gamma_2)$. Given $s=7$ and $c_{\rm R}=15$, the plot displays $I(X; \Gamma_1, \Gamma_2)$ with respect to the moduli $N_1$ and $N_2$ that are chosen to span the range $(0, 2^{15})$. In general, as $N_1$ or $N_2$ increases, $I(X; \Gamma_1, \Gamma_2)$ decreases and approaches zero. When both $N_1$ and $N_2$ are larger than $139$, corresponding to 99 percent of prime numbers in the range $(0, 2^{15})$, $I(X; \Gamma_1, \Gamma_2)$ quickly drops to $0.066$ bits. This indicates that information leakage in Montgomery domains can also be made negligibly small by choosing sufficiently large moduli values.

\begin{figure*}[t!]
  \centering  
  \subfigure[]{
    \includegraphics[width=.31\textwidth]{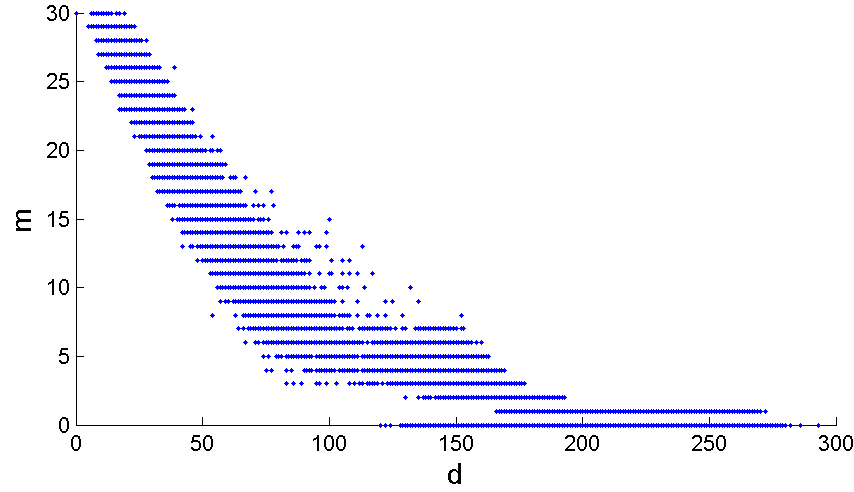} 
    \label{fig:anony:c}
  }  
  \subfigure[]{ 
    \includegraphics[width=.31\textwidth]{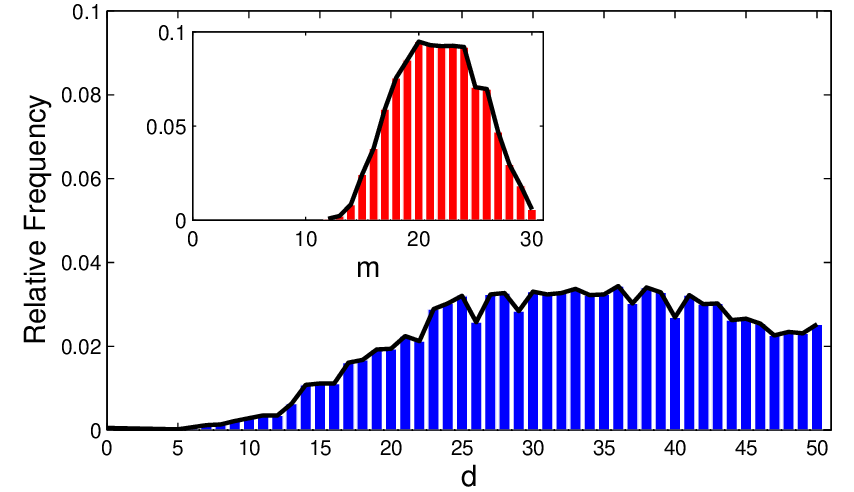} 
    \label{fig:anony:a}
  }
  \subfigure[]{ 
   \includegraphics[width=.31\textwidth]{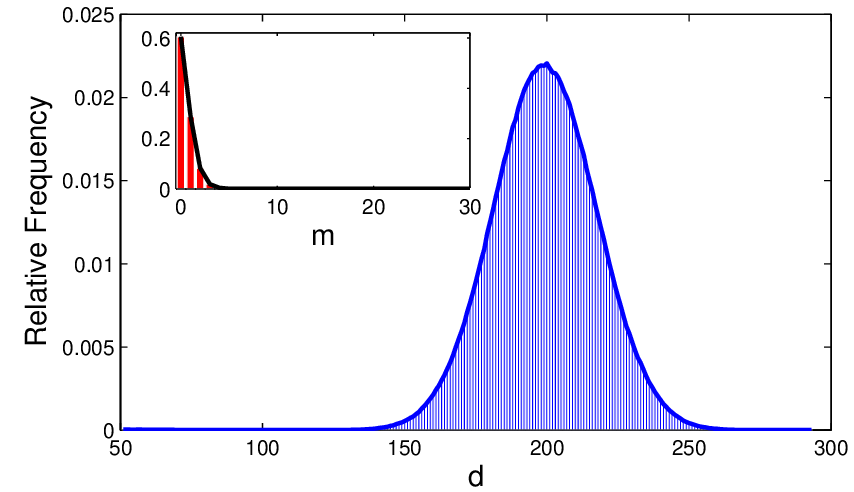} 
   \label{fig:anony:b}
  }
\caption{Effect of distance obfuscation with $L=30$ for the UBIRIS dataset. (a) Mapping from Hamming distance $d$ to obfuscated distance measure~$m$. (b) Histogram of $d$ and that of $m$ for neighbouring pairs with $d\leq 50$. (c) Histogram of $d$ and that of $m$ for non-neighbouring pairs with $d > 50$.} 
\label{fig:anony}
\end{figure*}   

\subsection{Privacy Gain in Index Generation} 
\label{sec:analysis:gain}

To measure the effect of applying privacy-preserving mechanisms in index generation, we introduce the notion of \emph{privacy gain}, defined in our context as the increase in the uncertainty of inferring a biometric identity $\mathbf{p}$ from the search indexes. 
Let $\hat{\mathbf{p}}$ be a random variable whose possible values are estimates of $\mathbf{p}$, which we recall is in the form of a binary string of length $D$ bits. In an information-theoretic approach, the privacy gain can be measured by the entropy of $\hat{\mathbf{p}}$ conditional on the knowledge of the search indexes. Accordingly,
the privacy gain is \emph{null} if one directly uses the segments $\{ \mathbf{p}^{(1)}, \mathbf{p}^{(2)}, \dots, \mathbf{p}^{(L)}\}$ for indexing $\mathbf{p}$ as the adversary can recover $\mathbf{p}=\mathbf{p}^{(1)}||\mathbf{p}^{(2)}||\dots||\mathbf{p}^{(L)}$ unambiguously. On the other hand, the privacy gain is \emph{maximum}, i.e., $D$ bits, if the search indexes are generated in an uninformative way.

Let us consider what if we use exactly the original form of the obfuscated segments $\{ \mathcal{V}_{\mathbf{p}^{(1)}}, \mathcal{V}_{\mathbf{p}^{(2)}}, \dots, \mathcal{V}_{\mathbf{p}^{(L)}} \}$ for hash-based indexing. It is clear in this case that, for any $i$, an estimate of $\mathbf{p}^{(i)}$ has only two outcomes in the set $\mathcal{V}_{\mathbf{p}^{(i)}}$, i.e., $\hat{\mathbf{p}}^{(i)}=\tilde{\mathbf{p}}^{(i)}_j$ for $j=0,1$ with equal probability. Thus, the conditional entropy $H(\hat{\mathbf{p}}^{(i)}\mid\mathcal{V}_{\mathbf{p}^{(i)}})$ is exactly one bit for all~$i$, and hence the privacy gain in this case is $L$ bits.

We recall that the baseline ``LSH + partial distance'' approach \cite{Weng16tkde} generates the search indexes by performing random sampling of bit positions. As a result, the more the distinct positions that are sampled, the less is the uncertainty left in inferring the biometric identity. The privacy gain of the LSH-based approach is thus equal to the number of unknown bits. For comparison, we randomly sample an input string with $s=8$ bits at a time for $L=50$ times with replacement. We perform simulations over 100,000 input strings each of length 400 bits. The input bits are independently generated and each bit is 0 or 1 with equal probability. The results show that the privacy gain in such a context is 147 bits on average.

Now, consider our proposed index generation scheme using the Montgomery form of each substring $\tilde{\mathbf{p}}^{(i)}_j \in \mathcal{V}_{\mathbf{p}^{(i)}}$. In this context, if every $\tilde{\mathbf{p}}^{(i)}_j$ has two independently generated signatures in Montgomery domains, we have a total of four randomized signatures for indexing $\mathbf{p}^{(i)}$, denoted by $\Gamma_1^{(i)}, \Gamma_2^{(i)}, \Gamma_3^{(i)}, \Gamma_4^{(i)}$. In this case, let us again treat the natural number representation of $\hat{\mathbf{p}}^{(i)}$ as a discrete random variable $X$ with probability mass function $p(x) = 2^{-s}$ for all $x$ in the range $\mathcal{X}=[0, 2^{s})$. Then, in a similar manner to the way we calculate $H(X|\Gamma_1, \Gamma_2)$ in Section \ref{sec:analysis:entropy} through \eqref{eq:equivocation}-\eqref{eq:joint2}, and setting $s=8$, $c_{\rm R}=15$, $c_{\rm N}=15$, we obtain the mean value of the conditional entropy $H( \hat{\mathbf{p}}^{(i)}\mid \Gamma_1^{(i)}, \Gamma_2^{(i)}, \Gamma_3^{(i)}, \Gamma_4^{(i)})$ as 7.82 bits and therefore the privacy gain is 391 bits for $L=50$, which is rather close to $D=400$ bits in this case.

\subsection{Privacy-Utility Trade-off in Similarity Search}
\label{sec:analysis:tradeoff}

As discussed in Section \ref{sec:search}, the distance obfuscation mechanism in the MIMP scheme enables inference-based similarity search to make judicious test decisions based on the obfuscated distance measure without the need of directly evaluating the Hamming distance. For any particular pair of $\mathbf{p}$ and $\mathbf{q}$, given the number of substrings $L$ and the value of $m$ that is obtained from (\ref{eq:m}), the Hamming distance $d(\mathbf{p}, \mathbf{q})$ is concealed in an interval given by (\ref{eq:dm}) that depends on the value of $m$. This is analogous to the concept of data anonymization \cite{Bakken04}.

It is important to note that the parameter $L$ plays a critical role in the trade-off between privacy and utility (i.e., search accuracy) in this context. We observe that:
\begin{itemize}
\item 
In the extreme case where $L=D$, each substring has exactly one bit. Since $d(\mathbf{p}^{(i)}, \mathbf{q}^{(i)})$ in this case is either 0 or 1, by Proposition \ref{thm:collision}, we must have $C(\mathcal{V}_{\mathbf{p}^{(i)}}, \mathcal{V}_{\mathbf{q}^{(i)}}) = 1$ for all $i$ and hence $m=D$ regardless of the value of $d(\mathbf{p}, \mathbf{q})$. This case yields the maximum privacy achievable as all strings $\mathbf{p}\in\Omega$ have the same $m$ value and thus cannot be ranked. However, in this extreme case, the inference-based approach does not work as similarity search is reduced to a wild guess, representing a worst-case scenario for the search accuracy.

\item 
In the other extreme case where $L=1$, by Proposition \ref{thm:collision}, we must have $m=1$ for $d(\mathbf{p}, \mathbf{q})\leq 1$, $m\leq 1$ for $d(\mathbf{p}, \mathbf{q})=2$ and $m=0$ for $d(\mathbf{p}, \mathbf{q})>2$. That is, for any string $\mathbf{p}\in\Omega$, the $m$ value is either 0 or 1. We can determine for all strings with $m=1$ that their Hamming distance to the query is not greater than two. For those with $m=0$, they cannot be further ranked. Accordingly, if the Hamming distance of the true match is greater than two, which is most likely true in a biometric database, the inference-based search in this extreme case is also no better than a wild guess.
\end{itemize}
Intuitively, when $L$ is neither too large nor too small, the values of the obfuscated distance measure $m$ are more likely to span a wider range. As we have observed in Section \ref{sec:search:rank} for the UBIRIS dataset, this makes it more feasible to order the retrieved candidates in multiple ranks and hence increases the search accuracy of the inference-based approach.

As a result of Proposition \ref{thm:collision}, for a pair of $\mathbf{p}$ and $\mathbf{q}$ with a particular value of $m$, the actual Hamming distance $d(\mathbf{p}, \mathbf{q})$ may take any integer value from a certain range. On the other hand, if $1<L<D$, it is clear that different distributions of the mismatching bits between $\mathbf{p}$ and $\mathbf{q}$ can lead to different values of $m$ even if the Hamming distance $d(\mathbf{p}, \mathbf{q})$ retains the same. Such a one-to-many relationship between the Hamming distance and the obfuscated distance measure ensures in all cases more or less privacy protection for the inference-based approach.

Fig. \ref{fig:anony} illustrates the effect of distance obfuscation with $L=30$ for the UBIRIS dataset. As discussed in Section \ref{sec:search:mimp}, we consider a total of 3,840,000 pairwise comparisons. Among them, given the Hamming radius $r=50$, we have 6,026 neighbouring pairs and 3,833,974 non-neighbouring pairs. In Fig. \ref{fig:anony:c}, we plot the mapping from the Hamming distance $d$ to the obfuscated distance measure $m$ for all the 3,840,000 pairwise comparisons, which clearly exhibits a one-to-many relationship between $d$ and $m$. In Fig. \ref{fig:anony:a}, we present the histogram of $d$ and that of $m$ for all neighbouring pairs (i.e., those with $d\leq 50$). Likewise, the histograms for all non-neighbouring pairs (i.e., those with $d>50$) are presented in Fig. \ref{fig:anony:b}. In both cases, we observe that the distribution of $d$ is very different from that of $m$.

\section{Performance Evaluation}
\label{sec:perform}

In this section, we evaluate the performance of the proposed MIMP scheme in randomized Montgomery domains for inference-based similarity search. For ease of description, we shall hence call it as MIMP-RM for short. As discussed in Section \ref{sec:search:rank}, a simple implementation of $k$-NN search using MIMP-RM is to choose the top $k$ candidates with the largest values of the obfuscated distance measure $m$. We demonstrate that the accuracy of candidate retrieval using this simple and privacy-preserving approach is close to that of conventional similarity search based on explicit distance values, but the associated cost is significantly reduced compared to cryptographic methods. We implement the test algorithms in Matlab and run the experiments on a 3.4 GHz Intel\textsuperscript{\textregistered} machine. For all experiments, the parameters used in MIMP-RM to generate the search indexes are $T=2$, $c_{\rm N}=15$, and $c_{\rm R}= 15$.

\subsection{Impact of Mismatching Bit Distribution}
\label{sec:perform:chara}

Here, we investigate how the mismatching bit distribution can affect the retrieval performance of MIMP-RM. Specifically, we randomly generate a query set containing 200 binary strings of length $D=2800$ bits. For each query~$\mathbf{q}$, we randomly generate a database of 200 records in the following way. First, only ten of them are $r$-neighbours of the query $\mathbf{q}$ where $r=350$. Second, we consider three different distributions of the mismatching bits. That is, the positions of the mismatching bits are randomly chosen from the range $[1, D/i]$ for $i=1, 2, 4$. In this way, a smaller $i$ corresponds to more evenly distributed mismatching bits over the binary string. Given that we know the ground truth set $\mathcal{B}(\mathbf{q}; r)$ in this experiment, we perform MIMP-RM with $L=200$ on the simulated dataset and evaluate the quality of the top $k$ retrieved candidates.

Fig. \ref{Fig:errbit} plots the \emph{precision-recall} (PR) curves under the three different mismatching bit distributions. The PR curves are derived by varying $k$ from 1 to 30. The results demonstrate that both the precision and the recall can be significantly improved if the mismatching bits are more evenly distributed. 
Thus, in cases where the mismatching bits are less evenly distributed, to improve the retrieval performance based on matching piecewise binary codes, it is helpful to apply a random synchronized permutation to binary feature vectors before performing the search. 
The results also show that the precision is very high when the recall is low, corresponding to the cases where $k$ is small and indicating that the retrieved candidates are mostly relevant. This further demonstrates the effectiveness of inference-based $r$-neighbour detection where the retrieved candidates with large $m$ values have a high probability that they are $r$-neighbours of the query as discussed in Section \ref{sec:search:mimp}.

\begin{figure}[t!]  
\centering
\includegraphics[width=.8\linewidth]{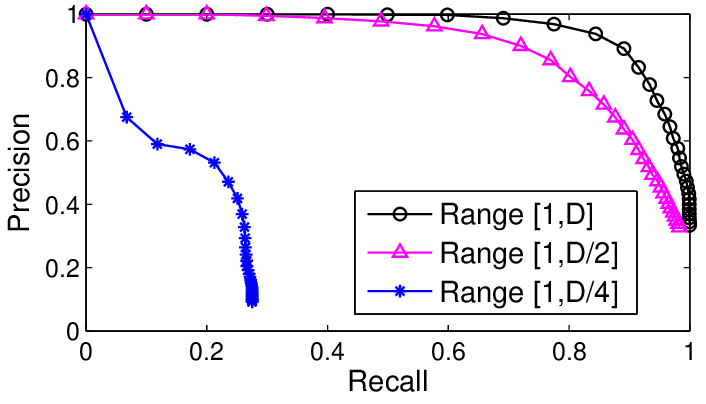}
\caption{Retrieval performance of MIMP-RM on the simulated dataset with positions of mismatching bits randomly chosen from a range.}
\label{Fig:errbit} 
\end{figure}

\subsection{Impact of the Parameter $L$}
\label{sec:perform:param}

In Section \ref{sec:analysis}, we argued that the number of substrings~$L$ is an important parameter of MIMP-RM. In particular, we showed in Section \ref{sec:analysis:entropy} that, to reduce the information leakage in randomized Montgomery domains, it is desirable to choose a large value of $s$, corresponding to a small value of $L$. In Section \ref{sec:analysis:tradeoff}, we established that choosing an appropriate value for $L$ is also subject to a fundamental trade-off between privacy and search accuracy. Here, we provide more extensive results and further demonstrate the impact of $L$ on the search accuracy of MIMP-RM. 

The experiments are conducted on public biometric datasets. In addition to the UBIRIS dataset, we also consider two benchmark face datasets. In particular, the LFW dataset \cite{LFW16} includes 13,233 face images crawled from the Internet for 5,749 subjects, and the FERET dataset \cite{FERET00} contains 2,400 face images of 200 subjects each with 12 images taken under semi-controlled environment. 
\begin{figure*} [t!]
\centering
\subfigure[]
{\includegraphics[width=.32\linewidth]{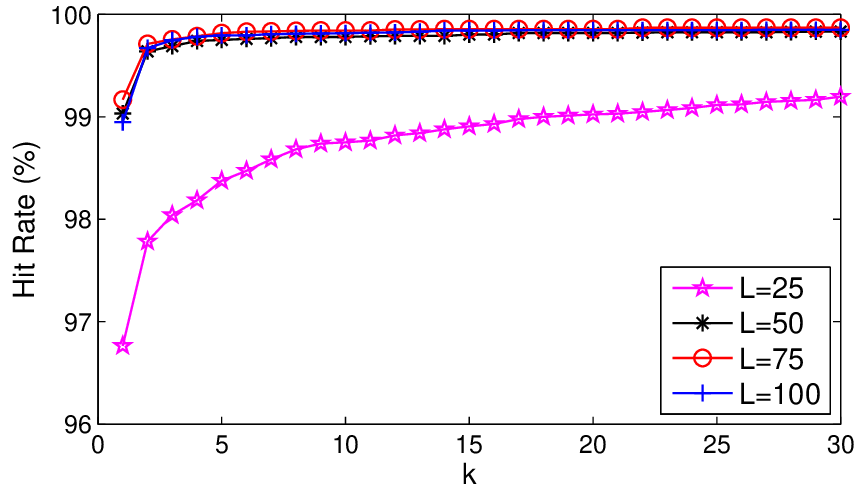}}
\subfigure[]
{\includegraphics[width=.32\linewidth]{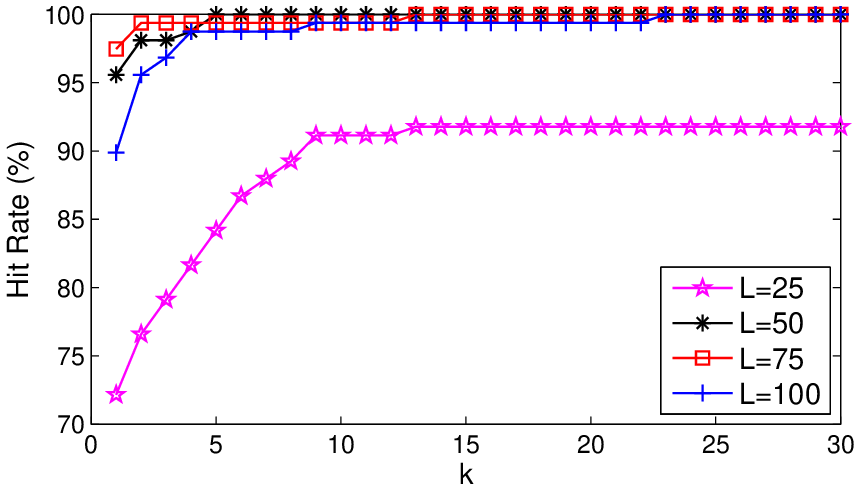}} 
\subfigure[]
{\includegraphics[width=.32\linewidth]{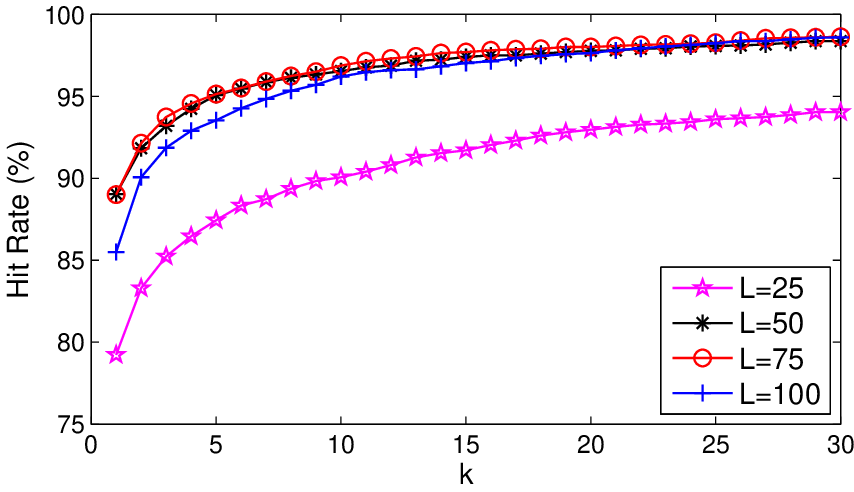}}
\caption{Impact of the parameter $L$ on the search accuracy of MIMP-RM with respect to the top $k$ retrieved candidates. (a) The UBIRIS dataset. (b) The LFW dataset. (c) The FERET dataset.}
\label{Fig:tablesize}
\end{figure*}

Nowadays, face recognition based on features learned from deep convolution networks, known as deep features, can achieve an accuracy over 99 percent on the once hardest LFW dataset \cite{Wang16face}. For the LFW dataset, we leverage the state-of-the-art deep features learned from a CNN framework \cite{Wu15DeepFace}. The deep network is trained on the CASIA-WebFace dataset \cite{Dong14deep}. Here, we use LSH embedding \cite{Andoni08} to convert the deep features into a binary string representation of 400 bits per LFW face image. To test our method with different data characteristics, we employ supervised PCA features on the FERET dataset. Specifically, we use six face samples per identity for training and six face samples for testing. We apply linearly separable subcodes \cite{Lim13} to convert the less robust PCA features into binary string representations of 448 bits per FERET face image, followed by a random synchronized permutation.

We enrol one sample per identity in the database and use another sample of that identity to search in each run of the experiments. The search accuracy is evaluated in terms of the \emph{hit rate}, defined as the percentage of queries with true match found in the top $k$ retrieved candidates. Fig. \ref{Fig:tablesize} presents the search accuracy results obtained by varying $k$ from 1 to 30. We test four different values of $L$, i.e., 25, 50, 75 and 100, respectively.
 
For all the three datasets, we observe that the hit rate increases significantly from $L=25$ to $L=50$ and drops at $L=100$, which is consistent with discussions in Section \ref{sec:analysis:tradeoff}. We also observe that the hit rate seems to peak at $L=75$ since there is only a marginal improvement from $L=50$ to $L=75$. To understand this latter effect, we perform an empirical study on the test samples of each database, and it reveals that over 99 percent of true match pairs have a Hamming distance smaller than $r^*=150$. Recall from Proposition \ref{thm:sic} that, in this case, with $L>r^*/2=75$, all such true match pairs are guaranteed to have a collision count and thus can be detected. Accordingly, in the context of MIMP-RM, choosing $L=r^*/2$ may serve as a rule of thumb for maximizing the search accuracy, and one may use a smaller $L$ value to strike a balance between search accuracy and privacy while reducing the information leakage of hash-based indexing in randomized Montgomery domains. 
 
\begin{figure*}[t]
\centering
\subfigure[] 
{\includegraphics[width=.32\textwidth]{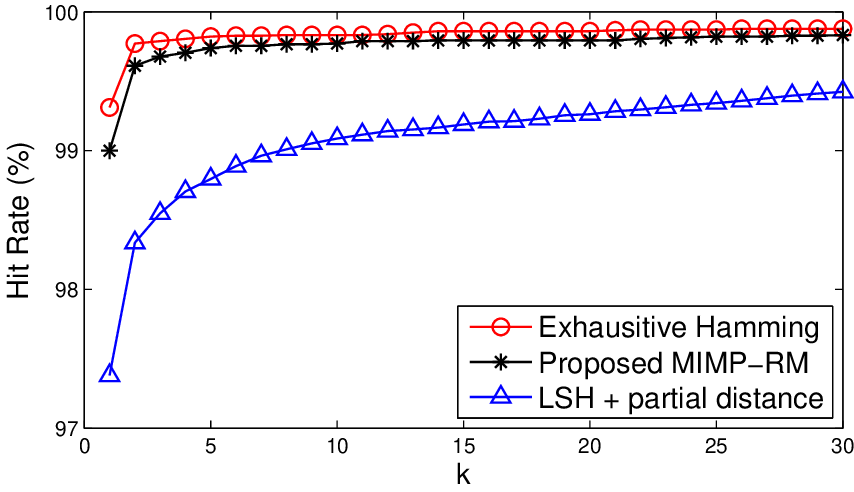} 
\label{fig:iris}}
\subfigure[]
{\includegraphics[width=.32\textwidth]{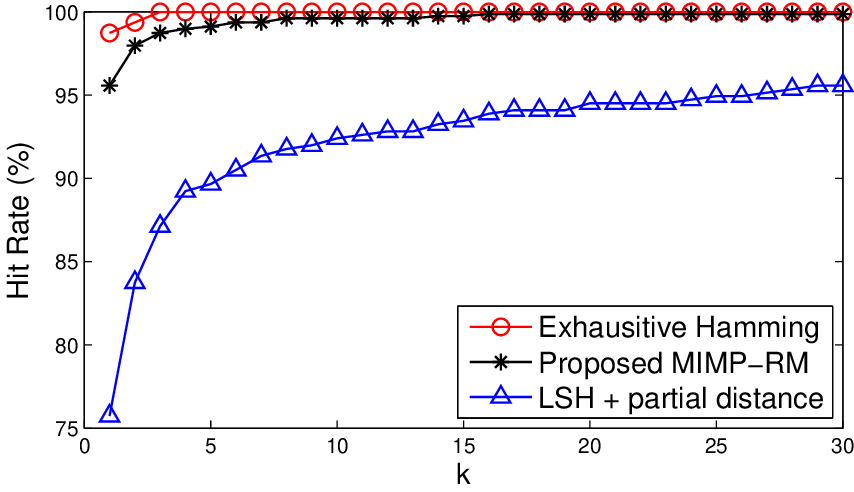} 
\label{fig:lfw}}
\subfigure[]
{
\includegraphics[width=.32\textwidth]{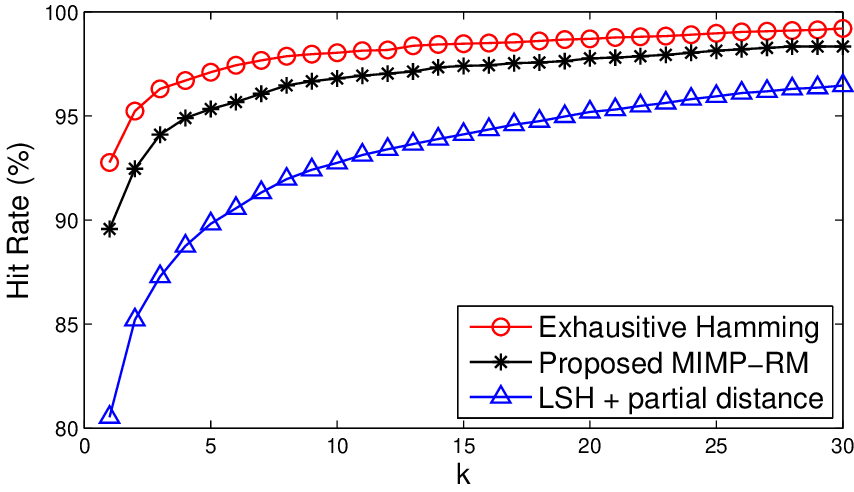} 
\label{fig:feret}}
\caption{Performance comparison in terms of the search accuracy with respect to the top $k$ retrieved candidates. (a) The UBIRIS dataset. (b) The LFW dataset. (c) The FERET dataset.}
\label{fig:search}
\end{figure*}

\subsection{Search Accuracy} 
\label{sec:perform:accuracy}

We demonstrate the effectiveness of MIMP-RM in search accuracy by comparing it with two baseline approaches that can be used in the context of privacy-preserving similarity search. In addition to the ``LSH + partial distance'' approach \cite{Weng16tkde}, we also consider what we call the ``exhaustive Hamming'' approach that performs exhaustive search by pairwise comparison based on full-string representations in Hamming space. Note that ``exhaustive Hamming'' provides the highest search accuracy that is achievable with secure computation methods as the latter can retain one-to-one matching accuracy in the encrypted domain \cite{Wu14}. However, such cryptography-based methods require pairwise comparison in the encrypted domain and have high computation cost as will be shown in Section \ref{sec:perform:cost}.

For ``LSH + partial distance'', we generate the unencrypted part by random sampling of bit positions as in Section \ref{sec:analysis:gain}. In this way, we create $L$ hash tables each being indexed by sub-hash codes of $s$ bits. To increase the ambiguity of query information for the server \cite{Weng15tifs}, each sub-hash code of the query has one bit value omitted so that the server needs to fill the absent bits. The retrieved candidates from the hit buckets are then ranked in the ascending order of their partial distance computed from the plaintext bit positions, i.e., the unencrypted part of the full-string representation.

Fig. \ref{fig:search} plots the search accuracy of the comparing methods in terms of the hit rate with respect to the $k$ value. The parameter $L$ is set to 50, 50 and 56 for the  UBIRIS dataset, the LFW dataset and the FERET dataset, respectively, which results in $s=8$ in all cases. We observe that MIMP-RM performs closely to ``exhaustive Hamming'', especially in Figs. \ref{fig:iris} and \ref{fig:lfw} where the full-string representations are highly discriminative in Hamming space. In all cases, MIMP-RM outperforms ``LSH + partial distance'' significantly. We believe it is mainly because the LSH-based approach in general has a low recall especially when the number of random hash functions is relatively small. In the above setting, the average number of plaintext bit positions is 233, 233 and 261 for the three datasets, respectively, representing about 58 percent of the full-string length. In contrast, MIMP-RM utilizes the full-string binary feature but in an obfuscated way that can provide highly accurate search results while increasing the uncertainty for adversarial learning in randomized Montgomery domains.

\begin{table*}[t] 
\centering
\caption{Comparison of computation cost for encoding piecewise binary codes generated from the UBIRIS dataset}
\label{table:costs} 
\renewcommand{\arraystretch}{1.2}
\begin{tabular}[c]{|l|c|c|c|c|}
\hline
 & Ciphertext size (MB) & Zipped file size (MB) & Encoding time (sec) & Scanning time (sec) \\
\hline
Paillier homomorphic encryption &  24  &  3.8   & 168.98  & 301.03  \\
\hline
Cryptographic hash function (SHA-256)  &  3  &  0.437   & 81.72  & 140.26   \\ 
\hline
Nested Montgomery form   &  0.155  &  0.029  & 1.01  & 0.0045 \\
\hline 
\end{tabular}
\end{table*}

\subsection{Computation Cost in Montgomery Domains}
\label{sec:perform:cost} 
  
Index generation and search in our proposed MIMP-RM approach mainly involve computing nested Montgomery forms. As discussed in Section \ref{sec:index:mont}, the Montgomery form of a natural number can be evaluated efficiently via (\ref{eq:MR_p}). To demonstrate the efficiency of computation in Montgomery domains, we compare it with two popular cryptographic methods, namely Paillier homomorphic encryption and the cryptographic hash function of SHA-256. We implement the former with the homomorpheR library and the latter with the sodium library on an R platform.

Table \ref{table:costs} provides the computation cost of each method conducted in R/3.2.2. The results are obtained for encoding piecewise binary codes that are generated from the UBIRIS dataset. The resulting cipher codes are further compressed using ZIP. We evaluate the computation cost by four measures. In particular, ``Ciphertext size'' and ``Zipped file size'' refer to the file size before and after compression, respectively. ``Encoding time'' refers to the CPU time required for encoding all records in the database. ``Scanning time'' refers to the CPU time required for a linear scan of the database for pairwise matching.

It can be seen in Table \ref{table:costs} that randomized signatures generated in Montgomery domains have much smaller code size and higher processing speed. Using the nested Montgomery forms, encoding is over 150 times faster and scanning is up to five orders of magnitude faster than that using Paillier homomorphic encryption. The results confirm that the efficiency of computation is significantly improved when done in Montgomery domains.

\section{Conclusion}
\label{sec:smry}

We have studied the problem of privacy-preserving similarity search in a biometric database to withstand adversarial machine learning based on the critical biometric similarity information. Unlike existing privacy protection methods for biometric identification that are in general required to perform cumbersome distance comparisons in the encrypted domain, the new approach proposed in this paper is based on statistical inference and carefully designed data obfuscation mechanisms that obviate the need for comparing exact distance values. We have also proposed to protect the data structures by performing hash-based indexing in randomized Montgomery domains with virtually negligible information leakage. Experiments on public biometric datasets have confirmed that candidate retrieval using our simple and privacy-preserving approach is both statistically reliable and computationally efficient. In future work, we are interested in extending the inference-based approach to other biometric modalities such as fingerprints that do not have a fixed-length binary string representation. We are also interested in extending the inference-based approach to search problems with distance measures other than in Hamming space.

% use section* for acknowledgment
\ifCLASSOPTIONcompsoc
  % The Computer Society usually uses the plural form
  \section*{Acknowledgments}
\else
  % regular IEEE prefers the singular form
  \section*{Acknowledgment}
\fi

The authors wish to thank the anonymous reviewers for their valuable comments that contributed to the improved quality of this paper. This work is supported by the Research Grants Council of Hong Kong (HKBU12202214) and the National Natural Science Foundation of China (61403324 and 61502058). J. Guo is the corresponding author. Most of the work was done while Y. Wang was with Hong Kong Baptist University.

% Generated by IEEEtran.bst, version: 1.14 (2015/08/26)
\bibliographystyle{IEEEtran}

\end{document}